\pdfoutput=1
\documentclass[9pt,nocopyrightspace]{sigplanconf}

\usepackage{amssymb,amsmath,amsthm,stmaryrd,bussproofs,mathrsfs,galois,picins}
\usepackage[linesnumbered,noend]{algorithm2e}
\usepackage{graphicx}
\usepackage{framed}
\usepackage{shuffle}
\usepackage{mathpartir}
\usepackage{wrapfig}
\usepackage{subfigure,multirow,arydshln}
\usepackage{xcolor}
\usepackage{hyperref}
\usepackage{thmtools,thm-restate}

\definecolor{mygrey}{rgb}{0.4,0.4,0.4}
\definecolor{myblue}{rgb}{0.2,0.2,0.7}
\definecolor{mygreen}{rgb}{0.2,0.7,0.2}
\definecolor{myred}{rgb}{0.8,0.2,0.2}

\usepackage{tikz}
\usetikzlibrary{positioning}
\usetikzlibrary{chains}
\usetikzlibrary{decorations}
\usetikzlibrary{backgrounds}
\usetikzlibrary{calc}
\usetikzlibrary{matrix}
\usetikzlibrary{scopes}
\usetikzlibrary{decorations.pathmorphing}
\usetikzlibrary{decorations.shapes}
\usetikzlibrary{shapes}
\usetikzlibrary{shapes.symbols}
\usetikzlibrary{decorations.pathreplacing}
\usetikzlibrary{fit}
\usetikzlibrary{backgrounds}

\DeclareFontFamily{U}  {MnSymbolC}{}
\DeclareFontShape{U}{MnSymbolC}{m}{n}{
    <-6>  MnSymbolC5
   <6-7>  MnSymbolC6
   <7-8>  MnSymbolC7
   <8-9>  MnSymbolC8
   <9-10> MnSymbolC9
  <10-12> MnSymbolC10
  <12->   MnSymbolC12}{}
\DeclareSymbolFont{MnSyC}{U}{MnSymbolC}{m}{n}
\DeclareMathSymbol{\righthalfcup}{\mathrel}{MnSyC}{184}

\theoremstyle{definition}

\newcommand{\romanqed}{$\righthalfcup$}

\declaretheoremstyle[
  spaceabove=5pt, spacebelow=5pt,
  headfont=\normalfont\bfseries,
  notefont=\mdseries, notebraces={(}{)},
  bodyfont=\normalfont,
  postheadspace=1em,
  qed=\romanqed
]{fkp}
\declaretheorem[style=fkp,numberwithin=section]{theorem}

\declaretheorem[style=fkp,sibling=theorem]{remark}
\declaretheorem[style=fkp,sibling=theorem]{example}

\declaretheorem[style=fkp,sibling=theorem]{definition}

\newcommand{\tuple}[1]{\langle #1 \rangle}

\newcommand{\true}{\textit{true}}
\newcommand{\false}{\textit{false}}
\newcommand{\hoare}[3]{\{{#1}\}\;\;#2\;\;\{{#3}\}}

\newcommand{\Loc}{\textsf{Loc}}

\newcommand{\until}{\textbf{U}}
\newcommand{\finally}{\textbf{F}}
\newcommand{\globally}{\textbf{G}}
\newcommand{\ltlnext}{\textbf{X}}

\usepackage{calc}

\usepackage{environ}
\newcounter{oldeqn}

\NewEnviron{constraints}{%
  \setcounter{oldeqn}{\value{equation}}%
  \addtocounter{oldeqn}{1}%
  \setcounter{equation}{0}%
  \begin{align}\BODY\end{align}%
  \setcounter{equation}{\value{oldeqn}}%
}

\tikzstyle{base} = [>=stealth,thick]
\tikzstyle{dfg} = [base,
  v/.style={rectangle,draw, minimum height=0.5cm, rounded corners=0.25cm}]

\makeatletter
\providecommand*{\revmodels}{%
  \mathrel{%
    \mathpalette\@revmodels\models
  }%
}
\newcommand*{\@revmodels}[2]{%
  \reflectbox{$\m@th#1#2$}%
}
\makeatother

\renewcommand\floatpagefraction{.9}
\renewcommand\topfraction{.9}

\newif\iflong
\longtrue

\AtBeginDocument{%
  \addtolength\abovedisplayskip{-0.6\abovedisplayskip}%
  \addtolength\belowdisplayskip{-0.6\belowdisplayskip}%
  \addtolength\abovedisplayshortskip{-0.55\abovedisplayshortskip}%
  \addtolength\belowdisplayshortskip{-0.55\belowdisplayshortskip}%
  \setlength{\textfloatsep}{2pt}
  \setlength{\intextsep}{4pt}
}

\usepackage{paralist}
\usepackage[small,tight]{bibhacks}
\newcommand{\rankformula}{w}
\newcommand{\rankformulas}{\mathscr{W}}
\newcommand{\closure}[1]{\langle\!\langle#1\rangle\!\rangle}

\SetCommentSty{mycommfont}
\newcommand{\idx}[1]{\text{\rm\texttt{#1}}}
\newcommand{\ic}[2]{{\tuple{#1 : #2}}}

\newcommand{\iv}[2]{#1(#2)}

\newcommand{\src}{\textsf{src}}
\newcommand{\tgt}{\textsf{tgt}}

\newcommand{\init}{\textsf{init}}

\newcommand{\lang}{\mathcal{L}}

\newcommand{\GV}{\mathsf{GV}}
\newcommand{\LV}{\mathsf{LV}}

\newcommand{\loc}{\mathsf{loc}}
\newcommand{\start}{\mathsf{start}}
\renewcommand{\phi}{\varphi}
\newcommand{\formulae}{\mathcal{F}}

\newcommand{\config}{\mathcal{C}}
\newcommand{\ctrans}[4]{#1 \xrightarrow{#2:#3} #4}

\newcommand{\ar}{\textit{ar}}

\newcommand{\ite}[3]{\texttt{if }#1\texttt{ then }#2\texttt{ else } #3}

\newcommand{\old}[1]{\textit{old}(#1)}
\newcommand{\ftrace}{finite trace}
\newcommand{\itrace}{infinite trace}
\newcommand{\natsleq}[1]{\{1,...,#1\}}
\newcommand{\iSigma}[1]{\Sigma(#1)}
\newcommand{\iVar}[1]{\textsf{Var}(#1)}
\newcommand{\iLV}[1]{\LV(#1)}
\newcommand{\QLTL}{\text{QLTL}}
\newcommand{\UP}{\textit{UP}}
\begin{document}

\title{Proving Liveness of Parameterized Programs}
\subtitle{(extended version)}
\authorinfo{Azadeh Farzan}
           {University of Toronto}{}
           {}
\authorinfo{Zachary Kincaid}
           {Princeton University}{}
           {}
\authorinfo{Andreas Podelski}
           {University of Freiburg}{}
           {}

\maketitle

\begin{abstract}
  Correctness of multi-threaded programs typically requires that they satisfy
  liveness properties.  For example, a program may require that no
  thread is starved of a shared resource, or that all threads eventually agree
  on a single value.  This paper presents a method for proving that such
  liveness properties hold.  Two particular challenges addressed in
  this work are that (1) the correctness argument may rely on global behaviour
  of the system (e.g., the correctness argument may require that all threads
  collectively progress towards ``the good thing'' rather than one thread
  progressing while the others do not interfere), and (2) such programs are
  often designed to be executed by \emph{any number} of threads, and the
  desired liveness properties must hold regardless of the number of threads that are active
  in the program.
\end{abstract}

\section{Introduction} \label{sec:intro}

Many multi-threaded programs are designed to be executed in parallel by an arbitrary
number of threads.  A challenging and practically relevant problem is to
verify that such a program is correct no matter how many threads are running.

Let us consider the example of the \emph{ticket mutual exclusion protocol}, pictured
in Figure~\ref{fig:ticket}.  This protocol is an idealized version of the one
used to implement spin-locks in the Linux kernel. The protocol maintains two
natural-typed variables: \texttt{s} (the \emph{service number}) and
\texttt{t} (the \emph{ticket number}), which are both initially zero.  A
fixed but unbounded number of threads simultaneously execute the protocol, which operates as follows.  First, the thread acquires a ticket by
storing the current value of the ticket number into a local variable
\texttt{m} and incrementing the ticket number (atomically).  Second, the thread waits for
the service number to reach \texttt{m} (its ticket value), and then enters
its critical section.  Finally, the thread leaves its critical section by
incrementing the service number, allowing the thread with the next ticket to
enter.

Mutual exclusion, a safety property, is perhaps the first property that comes
to mind for this protocol: no two threads should be in their critical sections
at the same time.  But one of the main reasons that the ticket protocol came
to replace simpler implementations of spin-locks in the Linux kernel was
because it satisfies \emph{non-starvation} \cite{ticket} (a liveness
property): no thread that acquires a ticket waits forever to enter its
critical section (under the fairness assumption that every thread is scheduled
to execute infinitely often).

Intuitively, the argument for non-starvation in the ticket protocol is
obvious: tickets are assigned to threads in sequential order, and whenever a
thread exits its critical section, the next thread in the sequence enters.
However, it is surprisingly difficult to come up with a formal correctness
argument manually, let alone automatically.  This paper presents a theoretical
foundation for algorithmic verification of liveness properties of
multi-threaded programs with any number of threads.

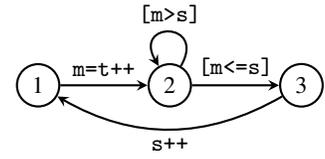
\begin{figure}
\begin{minipage}[b]{3.75cm}
  \noindent
\texttt{global nat s, t}\\
\texttt{local nat m}\\  
\textbf{while}\texttt{(true):}\\
\hspace*{0.25cm}\texttt{m=t++}\hfill{\textit{\color{gray}// Acquire a ticket}}\\
\hspace*{0.25cm}\textbf{while}\texttt{(m>s):}\hfill{\textit{\color{gray}// Busy wait}}\\
\hspace*{0.5cm}\texttt{skip}\\
\null\hfill{\textit{\color{red}// Critical section}}\\
\hspace*{0.25cm}\texttt{s++}\hfill{\textit{\color{gray}// Exit critical}}
\end{minipage}
\begin{minipage}[b]{4.5cm}
  \flushright
\begin{tikzpicture}[base,node distance=1.75cm]
  \node [circle,draw] (1) {1};
  \node [circle,draw,right of=1] (2) {2};
  \node [circle,draw,right of=2] (3) {3};
  \draw (1) edge[->] node[above]{\texttt{m=t++}} (2);
  \draw (2) edge[->] node[above]{\texttt{[m<=s]}} (3);
  \draw (3) edge[->,bend left] node[below]{\texttt{s++}} (1);
  \draw (2) edge[->,in=120,out=60,looseness=6] node[above]{\texttt{[m>s]}} (2);
\end{tikzpicture}
\end{minipage}
\caption{Ticket mutual exclusion protocol \label{fig:ticket}}
\end{figure}


The core of our method is the notion of \emph{well-founded proof spaces}.
Well-founded proof spaces are a formalism for proving properties of infinite
traces.  An \emph{infinite trace} is an infinite sequence of program commands
paired with thread identifiers, where a pair $\ic{\sigma}{\idx{i}}$ indicates
that the command $\sigma$ is executed by thread $\idx{i}$.  We associate with
each well-founded proof space a set of infinite traces that the space proves
to be terminating.  A well-founded proof space constitutes a proof of program
termination if every trace of the program is proved terminating.  A well-founded proof space
constitutes a proof of a liveness property if every trace of the program that
does \emph{not} satisfy the liveness property is proved terminating.


The main technical contribution of the paper is an approach to verifying that
a well-founded proof space proves that all program traces terminate.  Checking this condition is a language inclusion
problem, which is complicated by the fact that the languages consist of words
of infinite length, and are defined over an infinite alphabet (since each
command must be tagged with an identifier for the thread that executed it).
This inclusion problem
is addressed in two steps: first, we show how the inclusion between two sets
of infinite traces of a particular form can be proven by proving inclusion
between two sets of \emph{finite} traces
(Theorems~\ref{thm:up} and~\ref{thm:inclusion-soundness}).  This is essentially a reduction of
infinite trace inclusion to verification of a safety property; the reduction
solves the \emph{infinite length} aspect of the inclusion problem.  Second, we
develop \emph{quantified predicate automata}, a type of automaton suitable for
representing these languages that gives a concrete characterization of this
safety problem as an emptiness problem
(Theorem~\ref{thm:qpa_rec}).  In this context, quantification is used as a mechanism for
enforcing behaviour that \emph{all} threads must satisfy.  This solves the
\emph{infinite alphabet} aspect of the inclusion problem.

The overall contribution of this paper is a formal foundation for automating
liveness proofs for parameterized programs.  We investigate its theoretical
properties and pave the way for future work on exploring efficient algorithms
to implement the approach.



\subsection{Related work} \label{sec:relwork}

There exist proof systems for verifying liveness properties of parameterized
systems (for example, \cite{Sanchez2014}).  However, the problem of
automatically constructing such proofs has not been explored.  To the best of
our knowledge, this paper is the first to address the topic of automatic
verification of liveness properties of (infinite-state) programs with a
parameterized number of threads.

\textbf{Parameterized model checking} considers systems that consist of
unboundedly many finite-state processes running in parallel
\cite{Abdulla2010,journals/sttt/AbdullaJNdS12,PnueliRZ01,vmcai/FangPPZ04,tacas/FPPZ04,journals/corr/Durand-Gasselin15}.
In this paper, we develop an approach to the problem of verifying liveness
properties of parameterized \emph{programs}, in which processes are infinite
state.  This demands substantially different techniques than those used in
parameterized model checking.  The techniques used in this paper are more
closely related to \emph{termination analysis} and \emph{parameterized program
  analysis}.


\textbf{Termination analysis} an active field with many effective techniques
\cite{pldi/CookPR06,conf/tacas/CookSZ13,popl/CousotC12,conf/cav/HeizmannHP14,cav/LeeWY12,conf/tacas/Urban15}.
One of the goals of the present paper is to adapt the incremental style of
termination analysis pioneered by Cook et al. \cite{Cook2005,pldi/CookPR06} to the setting of parameterized
programs.  The essence of this idea is to construct a
termination argument iteratively via abstraction refinement: First, sample
some behaviours of the program and prove that those are terminating.  Second,
assemble a termination argument for the example behaviours into a candidate
termination argument.  Third, use a safety checker to prove that the
termination argument applies to \emph{all} behaviours of the program.  If the
safety check succeeds, the program terminates; if not, we can use the
counter-example to improve the termination argument.

Termination analyses have been developed for the setting of concurrent
programs \cite{Cook2007,Popeea2012,KetemaDonaldson2014}.  Our work differs in
two respects.  First, our technique handles the case that there are
unboundedly many threads operating simultaneously in the system.  Second, the
aforementioned techniques prove termination using \emph{thread-local}
arguments.  A thread-local termination argument expresses that each thread individually progresses towards some goal assuming that its
environment (formed by the other threads) is either passive or at least does
not disrupt its progress.  In contrast, the technique proposed in the
paper is able to reason about termination that requires coordination between all
threads (that is, all threads together progress towards some goal).  This
enables our approach to prove liveness for programs such as the Ticket
protocol (Figure~\ref{fig:ticket}): proving that some distinguished thread
will eventually enter its critical section requires showing that all
\emph{other} threads collectively make progress on increasing the value of the
service number until the distinguished thread's ticket is reached.

\textbf{Parameterized safety analysis} deals with proving safety properties of
infinite state concurrent programs with unboundedly many threads
\cite{Kaiser2014,Sanchez2012,Jaffar2009,Segalov2009}.  Safety analysis is
relevant to liveness analysis in two respects: (1)~In liveness analysis based
on abstraction refinement, checking the validity of a correctness argument is
reduced to the verification of a safety property \cite{Cook2005,pldi/CookPR06}
(2)~An invariant is generally needed in order to establish (or to
\emph{support}) a ranking function.  
Well-founded proof spaces can be seen as an extension of \emph{proof spaces}
\cite{Farzan2015}, a proof system for parameterized safety analysis, to prove
liveness properties.  A more extensive comparison between proof spaces and
other methods for parameterized safety analysis can be found in
\cite{Farzan2015}.


\section{Parameterized Program Termination} \label{sec:prelim}

This section defines parameterized programs and parameterized program
termination in a language-theoretic setting.


A \emph{parameterized program} is a multi-threaded program in which each
thread runs the same code, and where the number of threads is an input
parameter to the system.  A parameterized program can be specified by a
control flow graph that defines the code that each thread
executes.
A control flow graph is a directed, labeled graph \[P
= \tuple{\Loc,\Sigma,\ell_\init,\src,\tgt}\] where $\Loc$ is a set of program
locations, $\Sigma$ is a set of program commands, $\ell_\init$ is a designated
initial location, and $\src,\tgt : \Sigma \rightarrow \Loc$ are functions
mapping each program command to its source and target location.

Let $P$ be a program as given above. 
An \emph{indexed command}
$\ic{\sigma}{\idx{i}} \in \Sigma \times \mathbb{N}$ of $P$ is a pair
consisting of a program command $\sigma$ and an identifier $\idx{i}$ for the
thread that executes the command.\footnote{In the following, we will use typewriter font $\idx{i}$ as a meta-variable that ranges over thread identifiers (so \idx{i} is just a natural number, but one that is intended to identify a thread).}  For any natural number $N$, define
$\iSigma{N}$ to be the set of indexed commands $\ic{\sigma}{\idx{i}}$ with
$\idx{i} \in \natsleq{N}$.

Let $\Sigma$ be a set of program commands and $N \in \mathbb{N}$ be a natural
number.  A \emph{trace} over $\iSigma{N}$ is a finite or infinite sequence of
indexed commands.  We use $\iSigma{N}^*$ to denote the set of all finite traces over $\iSigma{N}$ and 
$\iSigma{N}^\omega$ to denote the set of infinite traces over $\iSigma{N}$.  For a finite trace $\tau$, we use $|\tau|$ to denote the length of
$\tau$.  For a (finite or infinite) trace $\tau$, we use 
$\tau_k$ to denote the $k^{\textit{th}}$ letter of $\tau$ and
$\tau[m,n]$ to denote the sub-sequence $\tau_m\tau_{m+1}\dotsi\tau_n$.  For a
finite trace $\tau$, and a (finite or infinite) trace $\tau'$, we use $\tau
\cdot \tau'$ to denote the concatenation of $\tau$ and $\tau'$.
We use
$\tau^\omega$  for the \itrace{} obtained by the infinite repeated
concatenation of the finite trace $\tau$
($\tau\cdot\tau\cdot\tau\dotsi$). 

For a parameterized program $P$ and a number $N \in \mathbb{N}$, we use $P(N)$
to denote a B\"{u}chi automaton that accepts the traces of the $N$-threaded
instantiation of the program $P$.  Formally, we define $P(N) =
\tuple{Q,\iSigma{N},\Delta,q_0,F}$ where
\begin{itemize}
\item $Q = \natsleq{N} \rightarrow \Loc$ (states are $N$-tuples of locations)
\item $\Delta = \{ (q,\ic{\sigma}{\idx{i}},q') : q(\idx{i}) = \src(\sigma) \land q' = q[\idx{i} \mapsto \tgt(\sigma)] \}$
\item $q_0 = \lambda i. \ell_\init$ (initially, every thread is at $\ell_\init$)
\item $F = Q$ (every state is accepting)
\end{itemize}
We use $\lang(P(N))$ to denote the language recognized by $P(N)$, and define
the set of traces of $P$ to be $\lang(P) = \bigcup_{N \in \mathbb{N}}
\lang(P(N))$.  We call the traces in $\lang(P)$ \emph{program traces}.

Fix a set of global variables $\GV$ and a set of local variables $\LV$.  For
any $N \in \mathbb{N}$, we use $\iLV{N}$ to denote a set of \emph{indexed
  local variables} of the form $\iv{l}{\idx{i}}$, where $l \in \LV$, and
$\idx{i} \in \natsleq{N}$.  $\iVar{N}$ denotes the set $\GV \cup \iLV{N}$.  We
do not fix the syntax of program commands.  A \emph{program assertion}
(program term) is a formula (term) over the vocabulary of some appropriate
theory augmented with a symbol for each member of $\GV$ and $\iLV{N}$ (for
all $N$).  For example, the program term (\texttt{x}(1) + \texttt{y}(2) +
\texttt{z}) refers to the sum of Thread 1's copy of the local variable
\texttt{x}, Thread 2's copy of the local variable \texttt{y}, and the global
variable \texttt{z}, and can be evaluated in 
a program state with at least the
threads $\{1,2\}$; the program assertion (\texttt{x}(1) $>$ \texttt{x}(2)) is
satisfied by any state (with at least the threads $\{1,2\}$) where Thread 1's
value for \texttt{x} is greater than Thread 2's.

We do not explicitly formalize the semantics of parameterized programs, but
will rely on an intuitive understanding of some standard concepts.  We write
$s \models \phi$ to indicate that the program state $s$ satisfies the program
assertion $\phi$.  We write $s \xrightarrow{\ic{\sigma}{\idx{i}}} s'$ to
indicate that $s$ may transition to $s'$ when thread $\idx{i}$ executes the
command $\sigma$.  Lastly, we say that a program state $s$ is initial if the
program may begin in state $s$.

A trace
\[ \ic{\sigma_1}{\idx{i}_1}\ic{\sigma_2}{\idx{i}_2}\dotsi \]
is said to be \emph{feasible} if there exists a corresponding infinite
execution starting from some initial state $s_0$:
\[ s_0 \xrightarrow{\ic{\sigma_1}{\idx{i}_1}} s_1 \xrightarrow{\ic{\sigma_2}{\idx{i}_2}} \dotsi \ . \]
A trace for which there is \emph{no} corresponding infinite execution is said to be \emph{infeasible}.

Finally, we may give our definition of parameterized program termination as follows:
\begin{definition}[Parameterized Program Termination]\label{def:parameterizedtermination}
  We say that a parameterized program $P$ \emph{terminates} if every program trace
of $P$
is \emph{infeasible}.  
That is, for
  every $N$, every $\tau \in \lang(P(N))$ is infeasible.
\end{definition}

This definition captures the fact that a counter-example to parameterized termination involves only finitely many threads (i.e., a counter example is a trace $\tau \in \lang(P(N))$ for some $N$).  
This is due to the definition of the set of traces of a parameterized program $\lang(P)$  (which is a language over an infinite alphabet) as an infinite union of languages $\lang(P(N))$, each over a finite alphabet.

The next two sections concentrate on parameterized program termination.  We will return to general liveness properties in Section~\ref{sec:liveness}.




\section{Well-founded Proof Spaces} \label{sec:proofspace}

A well-founded proof space is a formalism for proving parameterized
termination by proving that its set of program traces are infeasible.  This
section defines well-founded proof spaces, establishes a sound proof rule for
parameterized program termination, and describes how well-founded proof spaces
can be used in an incremental algorithm for proving parameterized program
termination.

\subsection{Overview}
We motivate the formal definitions that will follow in this section by
informally describing the role of well-founded proof spaces in an incremental
strategy (\emph{\'{a} la} \cite{Cook2005,pldi/CookPR06}) for proving
termination of parameterized programs.  The pseudo-code for this
(semi-)algorithm is given in Algorithm~\ref{alg:incr}.  The algorithm takes as
input a parameterized program $P$ and returns ``Yes'' if $P$ terminates,
``No'' if $P$ has a trace that can be proved non-terminating, and ``Unknown''
if the algorithm encounters a trace it cannot prove to be terminating or
non-terminating.  (There is also a fourth possibility that the algorithm runs
forever, repeatedly sampling traces but never finding a termination argument that
generalizes to the whole program).

\begin{algorithm}
  \SetAlgoLined\DontPrintSemicolon
  \SetKwInOut{Input}{Input} \SetKwInOut{Output}{Output} \Input{Parameterized
    program $P$}

  $B \gets \emptyset$ \tcc*{Initialize the basis, $B$}
  \tcc{Has every program trace been proved infeasible?}
  \While{$\lang(P) \nsubseteq \omega(\closure{B})$}{
    \tcc{Sample a possibly-feasible trace}
    Pick $\tau \in \lang(P) \setminus \omega(\closure{B})$\;
    \Switch{$\textsf{FindInfeasibilityProof}(\tau)$}{
      \Case{Infeasibility proof $\Pi$}{
        Construct $B'$ from $\Pi$ so that $\tau \in \omega(\closure{B'})$\;
        $B \gets B + B'$
      }
      \Case{Feasibility proof $\overline{\Pi}$}{
        \Return{No} \tcc*{$P$ is non-terminating}
      }
      \Other{
        \Return{Unknown} \tcc*{Inconclusive}
      }
    }
  }
  \Return{Yes}\tcc*{$P$ is terminating}
  \caption{Incremental algorithm for parameterized program termination \label{alg:incr}}
\end{algorithm}

Algorithm~\ref{alg:incr} builds a well-founded proof space by repeatedly
sampling traces of $P$, finding infeasibility proofs for the samples, and then
assembling the proofs into a well-founded proof space.  More precisely, the
algorithm builds a \emph{basis} $B$ for a proof space, which can be seen as a
finite set of axioms that generates a (typically infinite) well-founded proof
space $\closure{B}$.  The well-founded proof space $\closure{B}$ serves as an
infeasibility proof for a set of traces, which is denoted
$\omega(\closure{B})$ (Definition~\ref{def:omega-h}).  The goal of the
algorithm is to construct a basis for a well-founded proof space that proves
the infeasibility of \emph{every} program trace (at line 2, $\lang(P) \subseteq
\omega(\closure{B})$): if the algorithm succeeds in doing so, then $P$ terminates.


We will illustrate the operation of this algorithm on the simple example
pictured in Figure~\ref{fig:incdec}.  The algorithm begins with an empty basis
$B$ (at line 1): the empty basis generates an empty well-founded proof space
$\closure{B}$ that proves infeasibility of an empty set of traces (i.e.,
$\omega(\closure{B}) = \emptyset$).  Since the inclusion $\lang(P) \subseteq
\omega(\closure{B})$ does not hold (at line 2), we sample (at line 3) a
\emph{possibly-feasible} program trace $\tau \in \lang(P) \setminus
\omega(\closure{B})$ (we delay the discussion of how to verify the inclusion
$\lang(P) \subseteq \omega(\closure{B})$ to Section~\ref{sec:check}).  Suppose
that our choice for $\tau$ is the trace pictured in
Figure~\ref{fig:lasso-ex}(a), in which a single thread (Thread 1) executes the
loop forever.  This trace is \emph{ultimately periodic}: $\tau$ is of the form
$\pi \cdot \rho^\omega$, where $\pi$ (the \emph{stem}) and $\rho$ (the
\emph{loop}) are finite traces.  Under reasonable assumptions (that we
formalize in Section~\ref{sec:lasso-cex}) we ensure that sample traces
(counter-examples to the inclusion $\lang(P) \subseteq \omega(\closure{B})$)
are ultimately periodic.  The importance of ultimate periodicity is two-fold:
first, ultimately periodic traces have a (non-unique) finite representation: a
pair of finite words $\tuple{\pi,\rho}$.  Second, ultimately periodic traces
correspond to a simple class of sequential programs, allowing
Algorithm~\ref{alg:incr} to leverage the wealth of techniques that have been
developed for proving termination
\cite{cav/BradleyMS05,vmcai/PodelskiR04,conf/atva/HeizmannHLP13} and
non-termination \cite{Gupta2008}.  The auxiliary procedure
\textsf{FindInfeasibilityProof} denotes an (unspecified) algorithm that uses
such techniques to prove feasibility or infeasibility of a given trace.

Suppose that calling $\textsf{FindInfeasibilityProof}$ on the sample trace
$\tau$ gives the infeasibility proof pictured in Figure~\ref{fig:lasso-ex}(b)
and (c).  The infeasibility proof has two parts.  The first part is an
\emph{invariance proof}, which is a Hoare proof of an inductive invariant
(\mbox{\color{mygreen}$\texttt{d}(1) > 0$}) that supports the termination
argument.  The second part is a \emph{variance proof}, which is a Hoare proof
that (assuming the inductive invariant holds at the beginning of the loop)
executing the loop causes the state of the program to decrease in some
well-founded order.  This well-founded order is expressed by the \emph{ranking
  formula} {\color{purple}$\old{\texttt{x}} > \texttt{x} \land
  \old{\texttt{x}} \geq 0$} (the post-condition of the variance proof).  This
formula denotes a (well-founded) binary relation between the state of the
program and its \emph{old} state (the program state at the beginning of the
loop) that holds whenever the value of \texttt{x} decreases and was initially
non-negative.  Since there is no infinite descending sequence of program
states in this well-founded order, the trace $\tau$ (which executes the loop
infinitely many times) is infeasible.

We use the termination proof for $\tau$ to construct a basis $B'$ for a
well-founded proof space (at line 6).  This is done by breaking the
termination proof down into simpler components: the Hoare triples that were
used in the invariance and variance proofs, and the ranking formula that was
used in the variance proof.  The basis $B'$ constructed from
Figure~\ref{fig:lasso-ex} is pictured in Figure~\ref{fig:basis}.  We then add
$B'$ to the incrementally constructed basis $B$ (at line 7) and begin the loop again,
sampling another possibly-feasible trace \mbox{$\tau' \in \lang(P) \setminus
\omega(\closure{B})$}.

The incremental algorithm makes progress in the sense that it never samples
the same trace twice: if $\tau$ is sampled at some loop iteration, then $\tau
\in \omega(\closure{B})$ for all future iterations.  But in fact,
$\omega(\closure{B})$ contains infinitely many other traces, whose termination
proofs can be derived from the same basic building blocks (Hoare triples and
ranking formulas) as $\tau$.  For example, $\omega(\closure{B})$
contains all traces of the form
\[\ic{\texttt{x=pos()}}{\idx{i}}\ic{\texttt{d=pos()}}{\idx{i}}\big(\ic{\texttt{[x>0]}}{\idx{i}}\ic{\texttt{x=x-d}}{\idx{i}}\big)^\omega \]
(all of which are, intuitively, infeasible for the same reason as $\tau$).
The essential idea is that new Hoare triples and ranking formulas can
be deduced from the ones that appear in the basis $B$ by applying some simple
inference rules.  The resulting collections of Hoare triples and ranking
formulas (which are closed under these inference rules) forms a well-founded
proof space $\closure{B}$.  Thus in Algorithm~\ref{alg:incr}, well-founded
proof spaces serve as a mechanism for \emph{generalizing} infeasibility
proofs: they provide an answer to the question \emph{given infeasibility
  proofs for a finite set of sample traces, how can we re-arrange the
  ingredients of those proofs to form infeasibility proofs for other traces?}

We will stop our demonstration of Algorithm~\ref{alg:incr} here, concluding
with a listing of the remaining Hoare triples that must be discovered by the
algorithm to complete the proof (that is, if those triples are added to the
basis $B$, then $\omega(\closure{B})$ contains $\lang(P)$):
\begin{center}
    $\hoare{\texttt{d}(1) > 0}{\ic{\texttt{x=pos()}}{2}}{\texttt{d}(1) > 0}$\\
    $\hoare{\texttt{d}(1) > 0}{\ic{\texttt{d=pos()}}{2}}{\texttt{d}(1) > 0}$\\
    $\hoare{\texttt{d}(1) > 0}{\ic{\texttt{[x>0]}}{2}}{\texttt{d}(1) > 0}$\\
    $\hoare{\texttt{d}(1) > 0}{\ic{\texttt{x=x-d}}{2}}{\texttt{d}(1) > 0}$\\
    $\hoare{\old{\texttt{x}} \geq 0}{\ic{\texttt{[x>0]}}{1}}{\old{\texttt{x}} \geq 0}$\\
    $\hoare{\old{\texttt{x}} \geq 0}{\ic{\texttt{x=x-d}}{1}}{\old{\texttt{x}} \geq 0}$\\
    $\hoare{\old{\texttt{x}} > \texttt{x}}{\ic{\texttt{[x>0]}}{1}}{\old{\texttt{x}} \geq 0}$\\
    $\hoare{\texttt{d}(1) > 0 \land \old{\texttt{x}} > \texttt{x}}{\ic{\texttt{x=x-d}}{1}}{\old{\texttt{x}} > \texttt{x}}$\ .
  \end{center}

The remainder of this section is organized as follows: in
Section~\ref{sec:formalization}, we give the formal definition of well-founded
proof spaces, and describe how a well-founded proof space proves infeasibility
of an infinite set of traces.  This section treats well-founded proof spaces
as a mathematical object, divorcing it from its algorithmic side.  In
Section~\ref{sec:lasso-cex}, we describe \emph{regular} well-founded proof
spaces, a restricted form of well-founded proof spaces.  The key result in
this section (Theorem~\ref{thm:up}) is that to prove parameterized program
termination, it is sufficient for a \emph{regular} proof space to prove that
the ultimately periodic traces of the program terminate.

\begin{figure}
  \begin{minipage}[b]{2.4cm}
    \texttt{global int x}\\
    \texttt{local int d}\\
    \texttt{x = pos()}\\
    \texttt{d = pos()}\\
      \texttt{while (x > 0):}\\
      \hspace*{0.25cm}\texttt{x = x - d}
  \end{minipage}
  \hfill
  \begin{minipage}[b]{6cm}
    \flushright
    \begin{tikzpicture}[base,node distance=1.75cm]
      \node [circle,draw] (1) {1};
      \node [circle,draw,right of=1] (2) {2};
      \node [circle,draw,right of=2] (3) {3};
      \node [circle,draw,right of=3] (4) {4};
      \draw (1) edge[->] node[above]{\texttt{x=pos()}} (2);
      \draw (2) edge[->] node[above]{\texttt{d=pos()}} (3);
      \draw (3) edge[->,bend left] node[above]{\texttt{[x>0]}} (4);
      \draw (4) edge[->,bend left] node[below]{\texttt{x=x-d}} (3);
    \end{tikzpicture}
  \end{minipage}
\caption{Decrement example, pictured along side its control flow graph.  The
  expression \texttt{pos()} denotes a non-deterministically generated positive
  integer, and the command \texttt{[x>0]} is an \emph{assumption}; its execution does not change the state of the program, but it can only proceed when \texttt{x} is greater than 0. \label{fig:incdec}}
\end{figure}
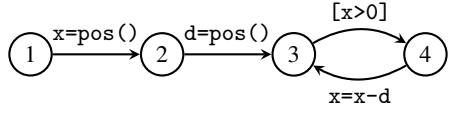

\begin{figure}
  \centering
  \subfigure[An ultimately periodic trace of Figure~\ref{fig:incdec}]{
    $\underbrace{\ic{\texttt{x=pos()}}{1}\ic{\texttt{d=pos()}}{1}}_{\text{Stem}}\hspace*{5pt}(\underbrace{\ic{\texttt{[x>0]}}{1}\ic{\texttt{x=x-d}}{1}}_{\text{Loop}})^\omega$
  }
  
  \subfigure[Invariance proof]{
    \begin{minipage}[b]{2.5cm}
      \begin{center}
        $\{ \true \}$\\
        $\ic{\texttt{x=pos()}}{1}$\\
        $\{ \true \}$\\
        $\ic{\texttt{d=pos()}}{1}$\\
        $\{ {\color{mygreen}\texttt{d}(1) > 0} \}$\\
        $\ic{\texttt{[x>0]}}{1}$\\
        $\{ \texttt{d}(1) > 0 \}$\\
        $\ic{\texttt{x=x-d}}{1}$\\
        $\{ {\color{mygreen}\texttt{d}(1) > 0} \}$
      \end{center}
    \end{minipage}
  }
  \hfill
  \subfigure[Variance proof]{
    \begin{minipage}[b]{5cm}
    \begin{center}
      $\{ {\color{mygreen}\texttt{d}(1) > 0} \land \old{\texttt{x}} = \texttt{x} \}$\\
      $\ic{\texttt{[x>0]}}{1}$\\
      $\{ \texttt{d}(1) > 0 \land \old{\texttt{x}} = \texttt{x} \land \old{\texttt{x}} \geq 0 \}$\\
      $\ic{\texttt{x=x-d}}{1}$\\
      $\{ {\color{purple}\old{\texttt{x}} > \texttt{x} \land \old{\texttt{x}} \geq 0} \}$
    \end{center}
    \end{minipage}
  }
  \caption{An ultimately periodic trace and termination proof. \label{fig:lasso-ex}}
\end{figure}

\begin{figure}
  \textbf{Hoare triples}:
  \begin{center}
      $\hoare{\true}{\ic{\texttt{x=pos()}}{1}}{\true}$\\
      $\hoare{\true}{\ic{\texttt{d=pos()}}{1}}{\texttt{d}(1) > 0}$\\
      $\hoare{\texttt{d}(1) > 0}{\ic{\texttt{[x>0]}}{1}}{\texttt{d}(1) > 0}$\\
      $\hoare{\texttt{d}(1) > 0}{\ic{\texttt{x=x-d}}{1}}{\texttt{d}(1) > 0}$\\
      $\hoare{\old{\texttt{x}} = \texttt{x}}{\ic{\texttt{[x>0]}}{1}}{\old{\texttt{x}} = \texttt{x}}$\\
      $\hoare{\old{\texttt{x}} = \texttt{x}}{\ic{\texttt{[x>0]}}{1}}{\old{\texttt{x}} \geq 0}$\\      
      $\hoare{\texttt{d}(1) > 0 \land \old{\texttt{x}} = \texttt{x}}{\ic{\texttt{x=x-d}}{1}}{\old{\texttt{x}} > \texttt{x}}$\\
      $\hoare{\old{\texttt{x}} \geq 0}{\ic{\texttt{x=x-d}}{1}}{\old{\texttt{x}} \geq 0}$
  \end{center}
  \textbf{Ranking formula}: $\old{\texttt{x}} > \texttt{x} \land \old{\texttt{x}} \geq 0$
    \caption{Basis computed from the termination proof in Figure~\ref{fig:lasso-ex} \label{fig:basis}}
\end{figure}

\subsection{Formal definition of Well-founded proof spaces} \label{sec:formalization}

A well-founded proof space is a set of Hoare triples and a set of ranking
terms, both closed under certain rules of inference.  They serve two roles.
First, they are the core of a proof rule for parameterized program
termination.  A well-founded proof space acts as a termination certificate for
a set of \itrace{}s (Definition~\ref{def:omega-h}); we may prove that a
program $P$ terminates by showing that all traces of $\lang(P)$ are contained
inside this set.  Second, well-founded proof spaces are a mechanism for
\emph{proof generalization}: starting from a (finite) \emph{basis} of Hoare
triples, we can take the closure of the basis under some simple inference
rules to form a well-founded proof space that proves the termination of a
larger set of traces (Definition~\ref{def:basis}).  We will now define these
notions formally.



We begin by formalizing the components of well-founded proof spaces,
\emph{Hoare triples} and \emph{ranking formulas}, and their inference rules.

A \emph{Hoare triple}
\[ \hoare{\phi}{\ic{\sigma}{\idx{i}}}{\psi} \]
consists of an indexed command $\ic{\sigma}{\idx{i}}$ and two program assertions $\phi$ and $\psi$ (the pre- and post-condition of the triple, respectively).
We say that such a triple is \emph{valid} if for any pair
of program states $s,s'$ such that $s \models \phi$ and $s
\xrightarrow{\ic{\sigma}{\idx{i}}} s'$, we have $s' \models \psi$.

We can infer new valid Hoare triples from a set of given ones using the
inference rules of proof spaces, namely \textsc{Sequencing},
\textsc{Symmetry}, and \textsc{Conjunction} \cite{Farzan2015}.  We will recall
the definition of these three rules below.

\textsc{Sequencing} is a variation of the classical sequencing
rule of Hoare logic.  For example, we may sequence the two triples
\[\hoare{\true}{\ic{\texttt{d=pos()}}{1}}{\texttt{d}(1) > 0}\text{ and}\]
\[\hoare{\texttt{d}(1) > 0}{\ic{\texttt{[x>0]}}{1}}{\texttt{d}(1) > 0} \]
to yield
\[\hoare{\true}{\ic{\texttt{d=pos()}}{1}\cdot\ic{\texttt{[x>0]}}{1}}{\texttt{d}(1) > 0}\ .\]

Two triples may be sequenced only when the post-condition of the first entails
the pre-condition of the first, according to a \emph{combinatorial entailment
  rule}.  The combinatorial entailment relation $\Vdash$ is defined as
\[
\phi_1 \land ... \land
\phi_n \Vdash \psi_1 \land ... \land \psi_m \mbox{ \ iff \ }\{ \phi_1,...,\phi_n \}
\supseteq \{\psi_1,..., \psi_m\}  
\]
(i.e., $\phi \Vdash \psi$ iff, viewed as sets of conjuncts, $\phi$ is a
superset of $\psi$).  Combinatorial entailment is a weaker version of logical entailment (which is used in the classical sequencing rule in Hoare logic).
Our sequencing rule can be written as follows:

\begin{mathpar}
  \inferrule[Sequencing]{
    \hoare{\phi_0}{\tau_0}{\phi_1}\\
    \phi_1 \Vdash \phi_1'\\
    \hoare{\phi_1'}{\tau_1}{\phi_2}
  }{
    \hoare{\phi_0}{\tau_0 \cdot \tau_1}{\phi_2}
  }
\end{mathpar}

\textsc{Symmetry} allows thread identifiers to be
substituted uniformly in a Hoare triple.  For example, from
\[\hoare{\true}{\ic{\texttt{d=pos()}}{1}}{\texttt{d}(1) > 0}\]
we may derive
\[\hoare{\true}{\ic{\texttt{d=pos()}}{2}}{\texttt{d}(2) > 0}\]
via the symmetry rule.  Given a permutation $\pi \in \mathbb{N} \rightarrow
\mathbb{N}$ and a program assertion $\phi$, we use $\phi[\pi]$ to denote the
result of substituting each indexed local variable $\iv{l}{\idx{i}}$ in $\phi$
with $\iv{l}{\pi(\idx{i})}$.  The symmetry rule may be written as follows:

\begin{mathpar}
  \inferrule[Symmetry]{
    \hoare{\phi}{\ic{\sigma_1}{\idx{i}_1}\dotsi\ic{\sigma_n}{\idx{i}_n}}{\psi} 
  }{  \hoare{\phi[\pi]}{\ic{\sigma_1}{\pi(\idx{i}_1)}\dotsi\ic{\sigma_n}{\pi(\idx{i}_n)}}{\psi[\pi]} 
  }
{\mbox{ \ 
    \begin{tabular}{l}
      $\pi  : \mathbb{N} \rightarrow \mathbb{N}$ \\[-4pt] is a permutation 
    \end{tabular}

}}
\end{mathpar}

\textsc{Conjunction} is precisely the conjunction
rule of Hoare logic.  For example, from the triples
\[\hoare{\texttt{d}(1) > 0}{\ic{\texttt{[x>0]}}{1}}{\texttt{d}(1) > 0}\text{ and} \]
\[\hoare{\old{\texttt{x}} = \texttt{x}}{\ic{\texttt{[x>0]}}{1}}{\old{\texttt{x}} \geq 0}\]
we may derive
\[\hoare{\texttt{d}(1) > 0 \land \old{\texttt{x}} = \texttt{x}}{\ic{\texttt{[x>0]}}{1}}{\texttt{d}(1) > 0 \land \old{\texttt{x}} \geq 0}\ .\]
The conjunction rule can be written as follows:

\begin{mathpar}
  \inferrule[Conjunction]{
    \hoare{\phi_1}{\tau}{\psi_1}\\
    \hoare{\phi_2}{\tau}{\psi_2}
  }{
    \hoare{\phi_1 \land \phi_2}{\tau}{\psi_1 \land \psi_2}
  }
\end{mathpar}

A \emph{proof space} is defined to be a set of valid Hoare triples that is
closed under these three rules \cite{Farzan2015}.  Proof spaces were used in
\cite{Farzan2015} to prove infeasibility of a set of \emph{finite} traces.  To
form a \emph{well-founded proof space}, which proves infeasibility of a set of
\emph{infinite} traces, we enrich a proof space with a set of \emph{ranking
  formulas}.

A ranking formula is a logical representation of a well-founded order on
program states.  We suppose that each program variable \texttt{x} has an
associated \emph{old} version $\textit{old}(\texttt{x})$ that allows formulas
to refer to the value of \texttt{x} in some ``previous'' state.  Any such
formula $\phi$ can be interpreted as a binary relation $R_\phi$ on states,
with $s \;R_\phi\; s'$ iff $\phi$ holds in the interpretation that uses $s$ to
interpret the \emph{old} variables and $s'$ to interpret the rest.  A
\emph{ranking formula} is defined to be a formula $w$ over the program
variables and their $\textit{old}$ copies such that the relation $R_w$ is a
well-founded order.

The only inference rule that we consider for ranking formulas is a
\emph{symmetry} rule: if $w$ is a ranking formula and $\pi : \mathbb{N}
\rightarrow \mathbb{N}$ is a permutation of thread identifiers, then
$w[\pi]$ is a ranking formula.

We may now define well-founded proof spaces formally:
\begin{definition}[Well-founded proof space]
  A \emph{well-founded proof space} $\tuple{\mathscr{H},\rankformulas}$ is a
  pair consisting of a set of Hoare triples $\mathscr{H}$ and a set of ranking
  formulas $\rankformulas$ such that $\mathscr{H}$ is closed under {\sc
    Sequencing}, {\sc Symmetry}, and {\sc Conjunction}, and $\rankformulas$ is
  closed under permutations of thread identifiers.
\end{definition}

We may present a well-founded proof as the closure of some \emph{basis}
(perhaps constructed from termination proofs of some small set of sample
traces).  Formally,
\begin{definition}  \label{def:basis}
Let $H$ be a set of valid Hoare triples, and let $W$ be a set of ranking
formulas.  $H$ and $W$ \emph{generate} a well-founded proof space
$\closure{H,W}$, defined to be the smallest well-founded proof space
$\tuple{\mathscr{H},\mathscr{W}}$ such that $H \subseteq \mathscr{H}$ and $W
\subseteq \mathscr{W}$.  We say that $\tuple{H,W}$ is a \emph{basis} for
$\tuple{\mathscr{H},\mathscr{W}}$.
\end{definition}
The fact that $\closure{H,W}$ is well-defined (i.e., contains only
\emph{valid} Hoare triples and ranking functions) follows immediately from the
soundness of the inference rules for well-founded proof spaces.

We associate with each well-founded proof space the set of all \itrace{}s that it proves infeasible.  Intuitively, a well-founded proof
space proves that a trace $\tau$ is infeasible by exhibiting a ranking formula
$w$ and a decomposition of $\tau$ into (infinitely many) finite segments
$\tau = \tau_0\tau_1\tau_2 \dotsi$ such that for each $i$, the program state
decreases in the order $w$ along each segment $\tau_i$.  More formally,
\begin{definition} \label{def:omega-h}
Let $\tuple{\mathscr{H},\rankformulas}$ be a well-founded proof space.  We define the set
$\omega(\mathscr{H},\rankformulas)$ of \itrace{}s recognized by
$\tuple{\mathscr{H},\rankformulas}$ to be the set of \itrace{}s $\tau =
\ic{\sigma_1}{\idx{i}_1} \ic{\sigma_2}{\idx{i}_2} \dotsi \in \Sigma(N)^\omega$ such that $\{
\idx{i}_k : k \in \mathbb{N} \}$ is finite and there exists a sequence of
naturals $\{ \alpha_k \}_{k \in \mathbb{N}}$, and a ranking formula $\rankformula \in \rankformulas$
such that:
\begin{enumerate}
\item For all $k \in \mathbb{N}$, $\alpha_k < \alpha_{k+1}$
\item For all $k \in \mathbb{N}$, there exists some formula $\phi$ such that
\[\hoare{\true}{\tau[1, \alpha_k]}{\phi} \in \mathscr{H}\text{, and}\]
\[\hoare{\phi \land \hspace*{-7pt}\bigwedge_{x \in \textsf{Var}(N)}\hspace*{-10pt} \old{x} = x}{\tau[\alpha_k+1,\alpha_{k+1}]}{\rankformula} \in \mathscr{H}   \let\qedsymbol\romanqed\qedhere\]
\end{enumerate}
\end{definition}

The fundamental property of interest concerning the definition of
$\omega(\mathscr{H},\rankformulas)$ is the following soundness theorem:

\begin{restatable}[Soundness]{theorem}{thmsoundness}
  \label{thm:soundness}
  Let $\tuple{\mathscr{H},\rankformulas}$ be a well-founded proof space.  Then every 
  \itrace{} in $\omega(\mathscr{H},\rankformulas)$ is infeasible.
\end{restatable}

Theorem~\ref{thm:soundness} is the basis of our proof rule for termination:
for a given program $P$, if we can derive a well-founded proof space
$\tuple{\mathscr{H},\rankformulas}$ such that
$\omega(\mathscr{H},\rankformulas)$ contains all the traces of $P$, then $P$
terminates.  This proof rule justifies the soundness of
Algorithm~\ref{alg:incr}.


\subsection{Ultimately periodic traces} \label{sec:lasso-cex}

Algorithm~\ref{alg:incr} relies on an auxiliary procedure
\textsf{FindInfeasibilityProof} to prove infeasibility of sample traces
(counter-examples to the inclusion $\lang(P) \subseteq \omega(\closure{B})$).
An attractive way of implementing \textsf{FindInfeasibilityProof} is to use  existing sequential program termination techniques
\cite{cav/BradleyMS05,vmcai/PodelskiR04,conf/atva/HeizmannHLP13} that are
very good at proving termination by synthesizing ranking functions for
programs of a restricted form, namely so-called \emph{lasso programs}.  To
take advantage of these techniques, we must ensure that the sample traces
given to \textsf{FindInfeasibilityProof} are \emph{ultimately periodic}, so
that they may be represented by lasso programs.  This section defines a
\emph{regularity} condition on well-founded proof spaces that enables us to
ensure that ultimately periodic counter-examples always exist.


An ultimately periodic trace is an \itrace{} of the form $\pi \cdot \rho^\omega$,
where $\pi$ and $\rho$ are finite traces.
Such a trace corresponds to a \emph{lasso program}, a sequential program that
executes the sequence $\pi$ followed by $\rho$ inside of a \textbf{while} loop
(since only finitely many threads are involved in $\pi$ and $\rho$, the local
variables of each thread may be renamed apart).

\parpic[r]{
  \begin{minipage}{2.75cm}
\texttt{x=0}\\
\textbf{while}\texttt{(true):}\\
\hspace*{0.25cm}\texttt{i = 0}\\
\hspace*{0.25cm}\textbf{while}\texttt{(i < x):}\\
\hspace*{0.5cm}\texttt{i++}\\
\hspace*{0.25cm}\texttt{x++}
\end{minipage}
} The question in this sub-section is \emph{how can we prove parameterized program
  termination while sampling only the counter-example traces to the sufficiency of the proof argument that are ultimately periodic}?  Phrased
differently, \emph{is proving termination of ultimately periodic traces enough
  to prove parameterized program termination?}  The (sequential) program to
the right illustrates the potential pitfall: even though every ultimately
periodic trace of the program is infeasible, the program does not terminate.

We place restrictions on well-founded proof spaces so that any (suitably restricted)
well-founded proof space that proves termination of all ultimately periodic
program traces inevitably proves termination of all program traces (i.e., if $\omega(\mathscr{H},\mathscr{W})$ includes all ultimately periodic
traces in $\lang(P)$, it inevitably contains all of $\lang(P)$).  These
restrictions are somewhat technical, and can be skipped on a first reading.

First, we exclude Hoare triples in
which local variables ``spontaneously appear'', such as \texttt{x}(2) in:
\[\hoare{\true}{\ic{\texttt{x = 0}}{1}}{\texttt{x}(1) = \texttt{x}(2) \lor \texttt{x}(1)=0}\]
This triple is valid, but the appearance of $\texttt{x}(2)$ in the post-condition is
arbitrary.  This technical restriction is formalized by well-formed Hoare
triples:
\begin{definition}[Well-formed Hoare triple] \label{def:wfht}
  A Hoare triple \[\hoare{\phi}{\tau}{\psi}\] is \emph{well-formed} if for each
  $\idx{i} \in \mathbb{N}$ such that an indexed local variable of the form
  $x(\idx{i})$ appears in the post-condition $\psi$, then either $\idx{i}$
  executes some command along $\tau$ or $x(\idx{i})$ or some other indexed
  local variable $y(\idx{i})$ with the same index $\idx{i}$ appears in the
  pre-condition $\phi$.
\end{definition}

The second restriction we make is to require that the well-founded proof space is generated by a finite basis in which there are no ``weak'' Hoare triples.  There are two types of weakness we prohibit. First, we exclude Hoare triples with conjunctive post-conditions
\[\hoare{\phi}{\tau}{\psi_1 \land \psi_2}\]
because such a triple can be derived from the pair 
\[\hoare{\phi}{\tau}{\psi_1} \text{ and } \hoare{\phi}{\tau}{\psi_2}\]
via the {\sc Conjunction} rule.  Second, we exclude Hoare triples for traces of length greater than one
\[\hoare{\phi}{\tau\cdot\ic{\sigma}{\idx{i}}}{\psi}\]
because such a triple can be derived from the pair
\[\hoare{\phi}{\tau}{\phi'} \text{ and } \hoare{\phi'}{\ic{\sigma}{\idx{i}}}{\psi}\]
(for some choice of $\phi'$) via the {\sc Sequencing} rule.  We formalize
these restrictions with \emph{basic} Hoare triples:
\begin{definition}[Basic Hoare triple]
  A Hoare triple
  \[ \hoare{\phi}{\ic{\sigma}{\idx{i}}}{\psi} \]
  is \emph{basic} if it is valid, well-formed, and
  the post-condition $\psi$ is atomic in the sense that it cannot be
  constructed by conjoining two other formulas.
\end{definition}


We call a well-founded proof space that meets all of these technical
restrictions \emph{regular}.  Formally:
\begin{definition}[Regular] \label{def:finitely-generated} 
  We say that a well-founded proof space $\tuple{\mathscr{H},\mathscr{W}}$ is
  \emph{regular} if there exists a \emph{finite} set of \emph{basic} Hoare
  triples $H$ and a \emph{finite} set of ranking formulas $W$ such that
  $\tuple{H,W}$ generates $\tuple{\mathscr{H},\mathscr{W}}$.
\end{definition}

The justification for calling such proof spaces regular is that if
$\tuple{\mathscr{H},\mathscr{W}}$ is regular, then
$\omega(\mathscr{H},\mathscr{W})$ is ``nearly'' $\omega$-regular, in the sense
that $\omega(\mathscr{H},\mathscr{W}) \cap \Sigma(N)^\omega$ is
$\omega$-regular (accepted by a B\"{u}chi automaton) for all $N \in
\mathbb{N}$.



Finally, we state the main result of this sub-section: regular
well-founded proof spaces guarantee the existence of
ultimately periodic counter-examples.  More precisely, if there is a sample
program trace that \emph{cannot} be proved terminating by a given regular
well-founded proof space, then there is also an \emph{ultimately
  periodic} counter-example.



\begin{theorem} \label{thm:up}
  Let $P$ be a parameterized program and let
  $\tuple{\mathscr{H},\rankformulas}$ be a regular well-founded proof space.  If
  every \emph{ultimately periodic}
  program trace $\pi \cdot\rho^\omega \in \lang(P)$ is included in
  $\omega(\mathscr{H},\rankformulas)$, then every program trace $\tau \in \lang(P)$ is included in
  $\omega(\mathscr{H},\rankformulas)$.
\end{theorem}
\begin{proof}
  For any set of traces $L \subseteq \bigcup_N\iSigma{N}^\omega$, we define
  $\UP(L)$ to be the set of ultimately periodic traces that belong to $L$:
\[ \UP(L) = \{ \tau \in L : \exists \pi,\rho \in \iSigma{N}^*. \tau = \pi  \rho^\omega) \}\ . \]

   We must prove that if $\UP(\lang(P)) \subseteq \omega(\mathscr{H},\rankformulas)$ then $\lang(P)) \subseteq \omega(\mathscr{H},\rankformulas)$.  We suppose that $\UP(\lang(P)) \subseteq \omega(\mathscr{H},\rankformulas)$ and (recalling the definition $\lang(P) = \bigcup_N \lang(P(N))$) prove that for all $N \in \mathbb{N}$, the inclusion $\lang(P(N)) \subseteq \omega(\mathscr{H},\rankformulas)$ holds.

   Let $N$ be an arbitrary natural number.
   Since $\UP(\lang(P(N))) \subseteq \UP(\lang(P)) \subseteq \omega(\mathscr{H},\rankformulas)$ we have
   \[ \UP(\lang(P(N))) \cap \Sigma(N)^\omega \subseteq \omega(\mathscr{H},\rankformulas) \cap \Sigma(N)^\omega\ .\]
   Since (by
  definition) $\lang(P(N)) \subseteq \Sigma(N)^\omega$, we can simplify:
  \begin{equation}
    \UP(\lang(P(N))) \subseteq \omega(\mathscr{H},\rankformulas) \cap \Sigma(N)^\omega\ .
    \label{eq:up}
  \end{equation}

  It is a well known fact that if $L_1$ and $L_2$ are $\omega$-regular
  languages (over a finite alphabet), then $\UP(L_1) \subseteq L_2$
  implies $L_1 \subseteq L_2$.  The language $\lang(P(N))$ is $\omega$-regular
  by definition, so if we can show that $\omega(\mathscr{H},\rankformulas) \cap
  \iSigma{N}^\omega$ is $\omega$-regular, then Inclusion~(\ref{eq:up})
  implies $\lang(P(N)) \subseteq \omega(\mathscr{H},\rankformulas) \cap \Sigma(N)^\omega$ and thus the desired result $\lang(P(N)) \subseteq
  \omega(\mathscr{H},\rankformulas)$.
 
  It remains only to show that $\omega(\mathscr{H},\rankformulas) \cap \iSigma{N}^\omega$
  is $\omega$-regular.  Here we will sketch the intuition \emph{why there
    exists} a B\"{u}chi automaton that recognizes $\omega(\mathscr{H},\rankformulas)
  \cap \iSigma{N}^\omega$.  Since $\tuple{\mathscr{H},\mathscr{W}}$ is regular, every
  Hoare triple in $\mathscr{H}$ is well-formed: this can be proved by
  induction on the derivation of the triple from the (well-formed) basis.  As a result, there
  are only finitely many program assertions that are relevant to the
  acceptance condition of a trace $\tau \in \iSigma{N}^\omega$ in
  $\omega(\mathscr{H},\rankformulas)$.  Intuitively, we can construct from this finite set
  of relevant assertions the finite state space of a B\"{u}chi automaton
  that recognizes $\omega(\mathscr{H},\rankformulas) \cap \iSigma{N}^\omega$.
\end{proof}

\textit{Discussion of Theorem~\ref{thm:up}.}
The example program above shows that it would not be sound to prove program
termination by proving termination of only its ultimately periodic program
traces.  However, it \emph{is} sound to check sufficiency of a candidate
regular well-founded proof space by inspecting only the ultimately periodic
program traces.  This soundness boils down to the fact that each infinite
execution involves only finitely many threads; more technically, the premise
of our proof rule (the inclusion between two sets of traces over an infinite
alphabet) is equivalent to the validity of an infinite number of inclusions
between $\omega$-regular languages over finite alphabets.

\section{Checking Proof Spaces} \label{sec:check}

The previous section defines a new proof rule for proving termination of
parameterized programs: given a parameterized program $P$, if there is some
well-founded proof space $\tuple{\mathscr{H},\rankformulas}$ such that
$\omega(\mathscr{H},\rankformulas)$ contains every trace of $P$, then $P$ terminates.
This section addresses two problems: (1) \emph{how can we verify that the
  premise of the proof rule holds?}, and (2) \emph{how can we generate an 
  ultimately periodic counter-example if it does not?}  The key idea in this section
is to reduce the problem of checking the premise (an inclusion problem for
sets of \itrace{}s over an unbounded alphabet) to a non-reachability
problem for a particular type of abstract machine (namely, \emph{quantified predicate automata}).

The first step in our reduction to non-reachability is to reduce the inclusion
$\lang(P) \subseteq \omega(\mathscr{H},\rankformulas)$ to an inclusion
problem on \ftrace{}s.  By Theorem~\ref{thm:up}, we know that it is sufficient
to check that the ultimately periodic traces of $\lang(P)$ are included in
 $\omega(\mathscr{H},\rankformulas)$.  Ultimately
periodic traces can be represented as finite traces which we call
\emph{lassos}.  A lasso is a finite trace of the form $\tau\$ \rho$, where
$\tau,\rho \in \iSigma{N}^*$ (for some $N$) and $\$$ is a special character not appearing in
$\iSigma{N}$.  A lasso $\tau\$\rho$ can be seen as a finite representation of
the ultimately periodic trace $\tau \cdot\rho^\omega$.  Note, however, that the
correspondence between lassos and ultimately periodic traces is not
one-to-one: an ultimately periodic trace $\tau\rho^\omega$ is represented by
infinitely many lassos, for example $\tau\$\rho, \tau\$\rho\rho,
\tau\rho\$\rho,$ and so on.

For a set of traces $L$, we define its lasso language as \[\$(L) = \{ \tau\$\rho
: \tau\cdot \rho^\omega \in L \} \] It is easy to show (using
Theorem~\ref{thm:up}) that the inclusion $\lang(P) \subseteq
\omega(\mathscr{H},\rankformulas)$ holds if and only if $\$(\lang(P)) \subseteq
\$(\omega(\mathscr{H},\rankformulas))$.  However, it is \emph{not} easy to give a direct
definition of $\$(\omega(\mathscr{H},\rankformulas))$ that lends
itself to recognition by an automaton of some variety.    Instead, we give an alternate lasso language $\$(\mathscr{H},\rankformulas)$
that is not exactly equal to $\$(\omega(\mathscr{H},\rankformulas))$, but (as we will see
in the following) is suitable for our purpose:

\begin{definition} \label{def:lasso-lang}
  Let $\tuple{\mathscr{H},\rankformulas}$ be a well-founded proof space.  Define
  $\$(\mathscr{H},\rankformulas)$ to be the set of lassos $\tau \$ \rho$ such that there
  is some $N \in \mathbb{N}$ so that $\tau,\rho \in \Sigma(N)^*$ and there
  exists some assertion $\phi$ and some ranking formula $\rankformula \in \rankformulas$ such that:
  \begin{enumerate}
  \item[i)] $\hoare{\true}{\tau}{\phi} \in \mathscr{H}$
  \item[ii)] $\hoare{\phi \land \bigwedge_{x \in \iVar{N}} \old{x} = x}{\tau}{\rankformula}
    \in \mathscr{H}$ \qedhere
  \end{enumerate}
\end{definition}

Note that $\$(\mathscr{H},\rankformulas)$ is neither a subset nor a superset of the set of
lassos $\$(\omega(\mathscr{H},\rankformulas))$ that correspond to ultimately periodic
words in $\omega(\mathscr{H},\rankformulas)$.  In fact, $\$(\mathscr{H},\rankformulas)$ may even
contain lassos $\tau\$\rho$ such that $\tau\cdot\rho^\omega$ is feasible:
consider for example the lasso $\ic{\texttt{y = 1}}{1} \$ \ic{\texttt{x = x -
    y}}{1}\ic{\texttt{y = -1}}{1}$: a well-founded proof space can prove
  that \texttt{x} decreases across the loop of this lasso, but this holds only
  for the \emph{first iteration} of the loop, and says nothing of subsequent
  iterations. Despite this, if the inclusion $\$(\lang(P)) \subseteq
  \$(\mathscr{H},\rankformulas)$ holds, then every trace of $P$ is proved infeasible by
  the well-founded proof space $\tuple{\mathscr{H},\rankformulas}$.  The intuition behind
  this argument is that if the inclusion $\$(\lang(P)) \subseteq
  \$(\mathscr{H},\rankformulas)$ holds, then for any ultimately periodic trace
  $\tau\cdot\rho^\omega$ of $\lang(P)$ \emph{every} representation of
  $\tau\cdot\rho^\omega$ as a lasso is included in $\$(\lang(P))$, and thus in
  $\$(\mathscr{H},\rankformulas)$.

\begin{theorem}[Inclusion Soundness] \label{thm:inclusion-soundness}
  Let $P$ be a parameterized program, and let $\tuple{\mathscr{H},\rankformulas}$ be a
  regular well-founded proof space.  If
  $\$(\lang(P)) \subseteq \$(\mathscr{H},\rankformulas)$, then $\lang(P) \subseteq
  \omega(\mathscr{H},\rankformulas)$.
\end{theorem}
\begin{proof}
  Suppose that the inclusion $\$(\lang(P)) \subseteq \$(\mathscr{H},\rankformulas)$ holds.
  By Theorem~\ref{thm:up}, it is sufficient to prove that every ultimately
  periodic trace of $\lang(P)$ is in $\omega(\mathscr{H},\rankformulas)$.  So let
  $\tau\cdot\rho^\omega \in \lang(P)$, and we will prove that $\tau\cdot\rho^\omega \in
  \omega(\mathscr{H},\rankformulas))$.

  Since $\tau\rho^\omega \in \lang(P)$, we must have $\tau\rho^n\$\rho^k \in
  \$(\lang(P)) \subseteq \$(\mathscr{H},\rankformulas)$ for all naturals $n$ and positive
  naturals $k$.  From the membership of $\tau\rho^n\$\rho^k$ in
  $\$(\mathscr{H},\rankformulas)$ and the definition of $\$(\mathscr{H},\rankformulas)$, there must
  exist some program assertion $\phi_{n,k}$ and some ranking formula
  $\rankformula_{n,k} \in \mathscr{W}$ such that:
  \[\hoare{\true}{\tau\rho^n}{\phi_{n,k}} \in \mathscr{H}\text{, and}\]
  \[\hoare{\phi_{n,k} \land \bigwedge_{x \in \iVar{N}}\hspace*{-5pt} \old{x} = x}{\rho^k}{\rankformula_{n,k}} \in \mathscr{H}\]

  Define an equivalence relation $\sim$ on the set of pairs $(n,m) \in
  \mathbb{N}^2$ such that $n<m$ by defining $(n,m) \sim (n',m')$ iff the
  ranking formulas $\rankformula_{n,m-n}$ and
  $\rankformula_{n',m'-n'}$ are equal.  Since the set of ranking
  formulas $\{ \rankformula_{n,k} \in \mathscr{W}: n,k \in \mathbb{N} \land k \geq 1
  \}$ is finite (following the same reasoning as in the proof of
  Theorem~\ref{thm:up}), the equivalence relation $\sim$ has finite index.  We
  use $[\rankformula]$ to denote the equivalence class consisting of all
  $(n,m)$ such that $\rankformula_{n,m-n} = \rankformula$.  By
  Ramsey's theorem \cite{Ramsey1930}, there is some ranking formula
  $\rankformula$ and some infinite set of naturals $D \subseteq \mathbb{N}$
  such that for all $d,d' \in D$ with $d<d'$, we have $(d,d') \in
  [\rankformula]$.

  We conclude that $\tau\rho^\omega \in \omega(\mathscr{H},\rankformulas)$ by observing
  (c.f. Definition~\ref{def:omega-h}) that there is an infinite sequence of
  naturals $\{\alpha_i\}_{i \in \mathbb{N}}$ defined by
  \[\alpha_i = |\tau| + d_i\cdot|\rho|\]
  (where $d_i$ is the $i^{\textit{th}}$ smallest element of $D$) such that the following hold:
  \begin{enumerate}
  \item[i)] For any $i \in \mathbb{N}$, since (by definition) $d_i<d_{i+1}$,
    we have
    \[ \alpha_i = |\tau|+d_i\cdot|\rho| < |\tau|+d_{i+1}\cdot|\rho| = \alpha_{i+1} \]
  \item[ii)] Let $i \in \mathbb{N}$, and define
    \begin{align*}
      n &= (\alpha_i - |\tau|)/|\rho|\\
      k &= (\alpha_{i+1} - \alpha_i)/|\rho|\ .
    \end{align*}
    Observe that:
    \begin{align*}
      \tau\rho^\omega[1,\alpha_i] &= \tau \cdot \rho^n\\
      \tau\rho^\omega[\alpha_i+1,\alpha_{i+1}] &= \rho^k\ .
    \end{align*}
    Recalling that
    $\rankformula = \rankformula_{n,k}$, it holds that

  \[\hoare{\true}{\tau\rho^n}{\phi_{n,k}} \in \mathscr{H}\]
  \[\hoare{\phi_{n,k} \land \bigwedge_{x \in \iVar{N}} \old{x} = x}{\rho^k}{\rankformula_{n,k}} \in \mathscr{H} \qedhere\]
  \end{enumerate}
\end{proof}

\begin{remark}
  We note that the reverse of the Inclusion Soundness theorem does not hold:
  if $\lang(P) \subseteq \omega(\mathscr{H},\rankformulas)$, it is not necessarily the
  case that $\$(\lang(P)) \subseteq \$(\mathscr{H},\rankformulas)$.
\end{remark}


\subsection{Quantified Predicate Automata}

The previous section establishes that a sufficient condition for verifying the
premise $\lang(P) \subseteq \omega(\mathscr{H},\rankformulas)$ of our proof
rule (an inclusion problem for sets of \itrace{}s) is to verify the inclusion
$\$(\lang(P)) \subseteq \$(\mathscr{H},\rankformulas)$ (an inclusion problem
for sets of \ftrace{}s).  In this section, we define \emph{quantified
  predicate automata}, a class of automata that are capable of recognizing the
difference $\$(\lang(P)) \setminus \$(\mathscr{H},\rankformulas)$.  This
allows us to characterize the problem of checking the inclusion $\$(\lang(P))
\subseteq \$(\mathscr{H},\rankformulas)$ as a safety problem: non-reachability
of an accepting configuration in a quantified predicate automaton (that is, the
\emph{emptiness} problem).

Quantified predicate automata (QPA) are infinite-state and recognize \ftrace{}s.  QPAs extend predicate automata
(\cite{Farzan2015}) with quantification, enabling them to recognize the
lasso language $\$(\lang(P))$.  Predicate automata are themselves
an infinite-state generalization of alternating finite automata
\cite{Brzozowski1980,Chandra1981}.  Our presentation of QPA will follow the
presentation of predicate automata from \cite{Farzan2015}.

Fix an enumeration $\{i_0,i_1,...\}$ of variable symbols.  Every
quantified predicate automaton is equipped with a \emph{finite relational
  vocabulary} $\tuple{Q,\ar}$, consisting of a finite set of predicate symbols
$Q$ and a function $\ar : Q \rightarrow \mathbb{N}$ that maps each predicate
symbol to its arity.  We use $\formulae(Q,\ar)$ to denote the set of positive
first-order formulas over the vocabulary $\tuple{Q,\ar}$, defined as follows:
\begin{align*}
  \phi,\psi \in \formulae(Q,\ar) ::= &\;\;\; q(i_{j_1},...,i_{j_{\ar(q)}}) \mid i_{j} = i_k \mid i_j \neq i_k\\
  & \mid \phi \land \psi \mid \phi \lor \psi \mid \forall i_j. \phi \mid \exists i_j. \phi
\end{align*}
Quantified predicate automata are defined as follows:
\begin{definition}[Quantified predicate automata]
  A \emph{quantified predicate automaton} (QPA) is a 6-tuple $A =
  \tuple{Q,\ar,\Sigma,\delta,\phi_\start,F}$ where
  \begin{itemize}
  \item $\tuple{Q,\ar}$ is a finite relational vocabulary,
  \item $\Sigma$ is a finite alphabet,
  \item $\phi_{\start} \in \formulae(Q,\ar)$ is a sentence over the vocabulary
    $\tuple{Q,\ar}$,
  \item $F \subseteq Q$ is a set of accepting predicate symbols, and
  \item $\delta : Q \times \Sigma \rightarrow \formulae(Q,\ar)$ is a
    transition function that satisfies the property that for any $q \in Q$ and
    $\sigma \in \Sigma$, the free variables of the formula $\delta(q,\sigma)$
    belong to the set $\{i_0,...,i_{\ar(q)}\}$.  The transition function
    $\delta$ can be seen as a symbolic rewrite rule
    \[ q(i_1,...,i_{\ar(q)}) \xrightarrow{\ic{\sigma}{i_0}} \delta(q,\sigma)\ , \]
    so the free variable restriction enforces that all variables on the
    right-hand-side are bound on the left-hand-side. \qedhere
  \end{itemize}
\end{definition}

A QPA $A = \tuple{Q,\ar,\Sigma,\delta,\phi_\start,F}$ defines an
infinite-state non-deterministic transition system, with transitions labeled
by indexed commands.  The configurations of the transition system are the set
of finite structures over the vocabulary $\tuple{Q,\ar}$.  That is, a
configuration $\config$ of $A$ consists of a finite universe $U^\config
\subseteq_{\text{fin}} \mathbb{N}$ (where $U^\config$ should be interpreted as
a set of thread identifiers) along with an interpretation $q^\config \subseteq
(U^\config)^{\ar(q)}$ for each predicate symbol $q \in Q$.  A configuration
$\config$ is \emph{initial} $\config \models \phi_\start$, and
\emph{accepting} if for all $q \notin F$, $q^\config = \emptyset$.  Given
$A$-configurations $\config$ and $\config'$, $\sigma \in \Sigma$, and $\idx{k}
\in U^\config$, $\config$ transitions to $\config'$ on reading
$\ic{\sigma}{\idx{k}}$, written $\ctrans{\config}{\sigma}{\idx{k}}{\config'}$,
if $\config$ and $\config'$ have the same universe ($U^\config =
U^{\config'}$), and for all predicate symbols $q \in Q$ and all
$\tuple{\idx{i}_1,...,\idx{i}_{\ar(q)}} \in q^\config$, we have
\[ \config' \models \delta(q,\sigma)[i_0 \mapsto \idx{k}, i_1 \mapsto \idx{i}_1,..., i_{\ar(q)} \mapsto \idx{i}_{\ar(q)}]\ . \]
For a concrete example of a transition, suppose that \[\delta(p,a) =
p(i_1,i_2) \lor (i_0 \neq i_1 \land q(i_2))\ .\] To make variable binding more
explicit, we will write this rule in the form \[\delta(p(i,j),
\tuple{a:k}) = p(i,j) \lor (k \neq i \land q(j))\ .\] For example,
if $\config$ is a configuration with $\config \models p(3,4)$, then a
transition $\ctrans{\config}{a}{1}{\config'}$ is possible only when ${\config' \models p(3,4) \lor (1 \neq 3 \land q(4))}$.


QPAs read input traces from right to left.  A trace \[\tau =
\ic{\sigma_1}{\idx{i}_1} \dotsi \ic{\sigma_n}{\idx{i}_n}\] is accepted by $A$
if there is a sequence of configurations $\config_{n},...,\config_{0}$ such
that $\config_{n}$ is initial, $\config_0$ is accepting, and for each $r \in
\{1,..., n\}$, we have
$\ctrans{\config_{r}}{\sigma_{r}}{\idx{i}_{r}}{\config_{r-1}}$.  We define
$\lang(A)$ to be the set all traces that are accepted by $A$.


Recall that the goal stated at the beginning of this section was to develop a
class of automaton capable of recognizing the difference $\$(\lang(P))
\setminus \$(\mathscr{H},\rankformulas)$ (for any given parameterized program $P$ and
regular well-founded proof space $\tuple{\mathscr{H},\rankformulas}$), and
thereby arrive at a sufficient automata-theoretic condition for checking the
premise of the proof rule established in Section~\ref{sec:proofspace}.  The following
theorem states that quantified predicate automata achieve this goal.
\begin{theorem} \label{thm:qpa_rec}
  Let $P$ be a parameterized program and a let $\tuple{\mathscr{H},\rankformulas}$ be a
  regular well-founded proof space.  Then
  there is a QPA that accepts $\$(\lang(P)) \setminus \$(\mathscr{H},\rankformulas)$.
\end{theorem}

The proof of this theorem proceeds in three steps: (I) $\$(\lang(P))$ is
recognizable by a QPA (Proposition~\ref{prop:pa-program}), (II)
$\$(\mathscr{H},\rankformulas)$ is recognizable by a QPA
(Proposition~\ref{prop:pa-program}), and (III) QPAs are closed under Boolean
operations (Proposition~\ref{prop:qpa-closure}).  Starting with step (I), we need the following proposition.
\begin{restatable}{proposition}{proppaprogram}
  \label{prop:pa-program}
  Let $P$ be a parameterized program.  Then there is a QPA $\mathcal{A}(P)$
  that accepts $\$(\lang(P))$.
\end{restatable}
\begin{proof}
  Let $P = \tuple{\Loc,\Sigma,\ell_\init,\src,\tgt}$ be a parameterized
  program.  For a word $\tau \in \iSigma{N}^*$ and a thread $\idx{i}$, define $\tau|_{\idx{i}}$ to be a the sub-sequence of $\tau$
  consisting of the commands executed by thread $\idx{i}$.  A word
  $\tau\$\rho$ is a lasso of $P$ if for each thread $\idx{i}$, (1)
  $\tau|_{\idx{i}}$ corresponds to a path in $P$, and (2) $\rho|_{\idx{i}}$
  corresponds to a loop in $P$.  We construct the QPA $\mathcal{A}(P) =
  \tuple{Q,\ar,\Sigma,\delta,\phi_\start,F}$ as follows:
  \begin{itemize}
  \item $Q = \Loc \cup (\Loc \times \Loc) \cup \{ \overline{\$} \}$, where
    $\$$ is a nullary predicate symbol and the rest are monadic.  The
    intuitive interpretation of propositions over this vocabulary are as
    follows:
    \begin{itemize}
    \item A trace is accepted starting from a configuration $\config$ with $\config \models \overline{\$}$ if the next letter
      is not $\$$.  This predicate is used to enforce the condition that the
      loop of a lasso is not empty.
    \item For each $\ell \in \Loc$ and each thread $\idx{i}$, a trace $\tau$
      is accepted starting from a configuration $\config$ with $\config \models \ell(\idx{i})$ if $\tau|_{\idx{i}}$
      corresponds to a path in $P$ ending at $\ell$.
    \item For each $\ell_1,\ell_2 \in \Loc$ and each thread $\idx{i}$, a trace
      $\tau\$\rho$ is accepted starting from a configuration $\config$ with $\config \models \tuple{\ell_1,\ell_2}(\idx{i})$
      if $\tau|_{\idx{i}}$ corresponds to a path in $P$ ending at $\ell_1$ and
      $\rho|_{\idx{i}}$ corresponds to a path in $P$ from $\ell_1$ to
      $\ell_2$.
    \end{itemize}
  \item $\phi_\start = \overline{\$} \land \forall i. \bigvee_{\ell \in \Loc}
    \tuple{\ell,\ell}(i)$
  \item $F = \{ \ell_\init \}$ (the automaton accepts when every thread
    returns to the initial location)
  \end{itemize}
  The transition function $\delta$ is defined as follows.

  For any location $\ell_1$, command $\sigma$, and thread $\idx{i}$, if
  $\tgt(\sigma) = \ell_1$ and $\idx{i}$ is at $\ell_1$, then reading
  $\ic{\sigma}{\idx{i}}$ causes thread $\idx{i}$ to move from $\ell_1$ to
  $\src(\sigma)$ while other threads stay put:
  \begin{align*}
    \delta(\tuple{\ell_1,\ell_2}(i), \ic{\sigma}{j}) &=
    (i = j \land \tuple{\tgt(\sigma),\ell_2}(i))\\
    &\hspace*{10pt}\lor (i \neq j \land \tuple{\ell_1,\ell_2}(i))\\
    \delta(\ell_1(i), \ic{\sigma}{j}) &= (i = j \land \tgt(\sigma)(i)) \lor (i \neq j \land \ell_1(i))
  \end{align*}
  For any location $\ell_1$, command $\sigma$, and thread $\idx{i}$, if thread
  $\idx{i}$ is at $\ell_1$ and $\ell_1 \neq \tgt(\sigma)$, then the automaton
  rejects when it reads $\ic{\sigma}{\idx{i}}$, but stays put when executing
  the command of another thread:
  \begin{align*}
    \delta(\tuple{\ell_1,\ell_2}(i), \ic{\sigma}{j}) &= (i \neq j \land \tuple{\ell_1,\ell_2}(i))\\
    \delta(\ell_1(i), \ic{\sigma}{j}) &= (i \neq j \land \ell_1(i))\ .
  \end{align*}
  Upon reading $\$$, the automaton transitions from $\tuple{\ell_1,\ell_1}(i)$
  to $\ell_1(i)$:
  \[
  \delta(\tuple{\ell_1,\ell_1}(i), \ic{\$}{j}) = \ell_1(i)\ ;
  \]
  but for $\ell_1 \neq \ell_2$, the automaton rejects:
  \[
  \delta(\tuple{\ell_1,\ell_2}(i), \ic{\$}{j}) = \false\ .
  \]
  Finally, the nullary predicate $\overline{\$}$ ensures that the next
  letter in the word is not $\$$:
  \begin{align*}
    \delta(\overline{\$}, \ic{\$}{j}) &= \false\\
    \delta(\overline{\$}, \ic{\sigma}{j}) &= \true\ . \qedhere
  \end{align*}

\end{proof}

Moving on to step (II):
\begin{restatable}{proposition}{proppaproof}
\label{prop:pa-proof}
   Let $\tuple{\mathscr{H},\rankformulas}$ be a regular well-founded proof space with basis $\tuple{H,W}$.  Then there is a QPA
   $\mathcal{A}(H,W)$ that accepts $\$(\mathscr{H},\rankformulas)$.
\end{restatable}

The construction is similar to the construction of a predicate automaton from a proof space \cite{Farzan2015}.  Intuitively, each Hoare triple in the basis of a regular proof
space corresponds to a transition of a QPA.  For example, the Hoare triple
\[\hoare{\texttt{d}(1) > 0 \land \old{\texttt{x}}=\texttt{x}}{\ic{\texttt{x=x-d}}{1}}{\old{\texttt{x}}>\texttt{x}}\]
corresponds to the transition
\[\delta([\old{\texttt{x}}>\texttt{x}], \ic{\texttt{x=x-d}}{i}) = [\texttt{d}(1)>0](i) \land [\old{\texttt{x}}=\texttt{x}]\ . \]
Details can be found in the appendix.

Finally, we conclude with step (III):
\begin{restatable}{proposition}{propqpaclosure}
\label{prop:qpa-closure}
  QPA languages are closed under Boolean operations (intersection, union, and
  complement).
\end{restatable}
The constructions follow the classical ones for alternating finite automata.  Again, details can be found in the appendix.


\subsection{QPA Emptiness}

We close this section with a discussion of the emptiness problem for
quantified predicate automata.  First, we observe that the emptiness problem
for QPA is undecidable in the general case, since emptiness is undecidable
even for quantifier-free predicate automata \cite{Farzan2015}.  In this
respect, our method parallels incremental termination provers for sequential
programs: the problem of checking whether a candidate termination
argument is sufficient is reduced to a safety problem that is undecidable.
Although the emptiness problem is undecidable, safety is a relatively
well-studied problem for which there are existing logics and algorithmic
techniques.  In particular, inductive invariants for QPA can serve as
certificates of their emptiness.  In the remainder of this section we
detail \emph{emptiness certificates}, which formalize this idea.


Intuitively, an \emph{emptiness certificate} for a QPA is a positive formula
that is entailed by the initial condition, inductive with respect to the
transition relation, and that has no models that are accepting
configurations.  A problem with this definition is that the transition
relation is infinitely-branching (we must verify that the emptiness
certificate is inductive with respect to the transition relation labeled with
any indexed command, of which there are infinitely many).  So first we define
a symbolic post-condition operator that gives a finite representation of these
infinitely many transitions.

Given a QPA $A = \tuple{Q,\ar,\Sigma,\delta,\phi_\start,F}$, we
define a \emph{symbolic post-condition} operator $\hat{\delta} :
\formulae(Q,\ar) \times \Sigma \rightarrow \formulae(Q,\ar)$ as follows:
\[\hat{\delta}(\phi, \sigma) = \exists i. \hat{\delta}(\phi, \ic{\sigma}{i}),\]
where $i$ is a fresh variable symbol not appearing in $\phi$ and
$\hat{\delta}(\phi, \ic{\sigma}{i})$ is the result of substituting each
proposition $q(j_1,...,j_{\ar(q)})$ that appears in $\phi$
with \[\delta(q,\sigma)[i_0 \mapsto i,i_1 \mapsto j_1, ..., i_{\ar(q)} \mapsto
  j_{\ar(q)}]\ .\]

We may now formally define emptiness certificates:
\begin{definition} \label{def:emp_cert}
  Let $A = \tuple{Q,\ar,\Sigma,\delta,\phi_\start,F}$ be a QPA.  An
  \emph{emptiness certificate} for $A$ is a positive first-order formula $\phi
  \in \formulae(Q,\ar)$ along with proofs of the following entailments:
  \begin{enumerate}
  \item \emph{Initialization}: $\phi_\start \vdash \phi$
  \item \emph{Consecution}: For all $\sigma \in \Sigma$, $\hat{\delta}(\phi,\sigma) \vdash \phi$
  \item \emph{Rejection}: $\phi \vdash \bigvee_{q \in Q\setminus F} \exists i_1,...,i_{\ar(q)}. q(i_1,...,i_{\ar(q)})$. \qedhere
  \end{enumerate}
\end{definition}



The following result establishes that that emptiness certificates are a sound
proof system for verifying emptiness of a QPA.
\begin{restatable}{theorem}{thmcert}
  \label{thm:cert}
    Let $A = \tuple{Q,\ar,\Sigma,\delta,\phi_\start,F}$ be a QPA.  If there is
    an emptiness certificate for $A$, then $\lang(A)$ is empty.
\end{restatable}

\section{Beyond Termination} \label{sec:liveness}
In the last two sections presented a technique that uses well-founded proof
spaces to prove that parameterized programs terminate.  This section extends
the technique so that it may be used to prove that parameterized programs
satisfy general liveness properties.  The class of liveness properties we
consider are those that are definable in \emph{(thread) quantified linear
  temporal logic} ($\QLTL$), which extends linear temporal logic with thread
quantifiers to express properties of parameterized systems.


Given a finite alphabet $\Sigma$, a $\QLTL(\Sigma)$ formula is built
using the connectives of first-order and linear temporal logic, where
quantifiers may \textbf{not} appear underneath temporal modalities, and where
every proposition is either $i = j$ (for some thread variables $i,j$) or
$\ic{\sigma}{i}$ (for some $\sigma \in \Sigma$ and thread variable $i$).  A satisfaction relation $\models$ defines when a trace $\tau$ satisfies a
$\QLTL(\Sigma)$ formula, using a map $\mu$ to interpret free variables:
\begin{align*}
  \tau,\mu \models i = j &\iff \mu(i) = \mu(j)\\
  \tau,\mu \models \ic{\sigma}{i} &\iff \tau_1 = \ic{\sigma}{\mu(i)}\\
  \tau,\mu \models \phi\until\psi &\iff \exists k \in \mathbb{N}. (\forall i<k. \tau[i,\omega],\mu \models \phi)\\
  &\hspace*{2cm}\land (\tau[k,\omega],\mu \models \psi)\\
  \tau,\mu \models \ltlnext\phi &\iff \tau[2,\omega] \models \phi\\
  \tau,\mu \models \exists i. \phi &\iff \exists \idx{i}\in\natsleq{N}.\tau,\mu[i \mapsto \idx{i}] \models \phi\\
  \tau,\mu \models \phi \land \psi &\iff \tau,\mu \models \phi \land \tau,\mu \models \psi\\
  \tau,\mu \models \lnot\phi &\iff \tau,\mu \not\models \phi
\end{align*}
The rest of the connectives are defined by the usual equivalences ($\forall
i.\phi \equiv \lnot\exists i.\lnot\phi$, $\mathbf{F}\phi \equiv
\true\mathbf{U}\phi$, $\mathbf{G}\phi \equiv \lnot\mathbf{F}\lnot\phi$, $\phi\lor \psi \equiv \lnot (\lnot \phi \land \lnot \psi)$).
For a concrete example, the following formula expresses the liveness property
of the ticket protocol (Figure~\ref{fig:ticket}), ``if every thread executes
infinitely often, then no thread is starved'':
\[\big(\forall i.\globally\finally\bigvee_{\sigma \in \Sigma} \ic{\sigma}{i}\big) \Rightarrow \big(\forall i.\globally\finally\ic{\texttt{[m<=s]}}{i}\big)\]


The theorem enabling well-founded proof spaces to verify
$\QLTL(\Sigma)$ properties is the following:
\begin{restatable}{theorem}{thmqltl} \label{thm:qltl}
  Let $\Sigma$ be a finite alphabet, and let $\phi$ be a $\QLTL(\Sigma)$
  sentence.  There is a QPA $\mathcal{A}(\phi)$ that recognizes
  the language:
  \[ \$(\lang(\phi)) = \{ \tau\$\rho \in \bigcup_N\Sigma(N)^\omega : \tau\rho^\omega \models \phi \}  \let\qedsymbol\romanqed\qedhere \]
\end{restatable}
\begin{proof}
  See appendix.
\end{proof}

This theorem allows us to employ a classical idea for temporal verification
\cite{Vardi1986}: to show that every execution of a program satisfies a
$\QLTL$ property $\phi$, we show that every program trace that \emph{violates}
$\phi$ is infeasible.  Thus, we have the following proof rule: given a
$\QLTL$ sentence $\phi$ and a parameterized program $P$, if there exists
regular well-founded proof space $\tuple{\mathscr{H},\mathscr{W}}$ with basis
$\tuple{H,W}$ such that the language $\lang(\mathcal{A}(P) \land
\mathcal{A}(\lnot\phi) \land \lnot \mathcal{A}(H,W))$ is empty, then $P$
satisfies $\phi$.

\begin{example}
To illustrate the idea behind Theorem~\ref{thm:qltl}, we give a manual
construction of a QPA for the (negated) liveness property of the ticket
protocol.  The negated liveness property can be written as a conjunction of a
fairness constraint and a negated liveness constraint:
\[\big(\forall i.\globally\finally\bigvee_{\sigma \in \Sigma} \ic{\sigma}{i}\big) \land \big(\exists i.\finally\globally\lnot\ic{\texttt{[m<=s]}}{i}\big)\]

Given a lasso $\tau\$\rho$, the ultimately periodic word $\tau\rho^\omega$
satisfies the above property iff each thread executes some
command along $\rho$ (left conjunct) and there is some thread that does
\emph{not} execute \texttt{[m<=s]} along $\rho$ (right conjunct).  We construct a QPA with two
monadic predicates $\textit{exec}$ and $\overline{\textit{enter}}$ and one
nullary predicate $\overline{\$}$ such that
\begin{itemize}
\item $\tau\$\rho$ is accepted from a configuration $\config$ with $\config
  \models \textit{exec}(\idx{i})$ iff $\rho$ contains a command of thread $\idx{i}$,
\item $\tau\$\rho$ is accepted from a configuration $\config$ with $\config
  \models \overline{\textit{enter}}(\idx{i})$ iff $\rho$ does not contain
  $\ic{\texttt{[m<=s]}}{\idx{i}}$, and
\item $\tau$ is accepted from a configuration $\config$ with $\config \models
  \overline{\$}$ iff $\tau$ does not contain $\ic{\$}{\idx{i}}$ for any thread
  $\idx{i}$.
\end{itemize}
The initial formula of the QPA is $(\forall i. \textit{exec}(i)) \land
(\exists i. \overline{\textit{enter}}(i))$ and the only accepting predicate
symbol is $\overline{\$}$.  The transition relation of the QPA is as follows:
\begin{align*}
  \delta(\overline{\textit{enter}}(i), \ic{\texttt{[m<=s]}}{j}) &= i \neq j \land \overline{\textit{enter}}(i)\\
  \delta(\overline{\textit{enter}}(i), \ic{\$}{j}) &= \overline{\$}\\
  \delta(\overline{\textit{enter}}(i), \ic{\texttt{m=t++}}{j}) &= \overline{\textit{enter}}(i)\\
  \delta(\overline{\textit{enter}}(i), \ic{\texttt{s++}}{j}) &= \overline{\textit{enter}}(i)\\
  \delta(\overline{\textit{enter}}(i), \ic{\texttt{[m>s]}}{j}) &= \overline{\textit{enter}}(i)\\
  \delta(\textit{exec}(i), \ic{\$}{j}) &= \false\\
  \delta(\overline{\$}, \ic{\$}{j}) &= \false
\end{align*}
and for all $\sigma \neq \$$,
\begin{align*}
  \delta(\textit{exec}(i), \ic{\sigma}{j}) &= i = j \lor \textit{exec}(i)\\
  \delta(\overline{\$}, \ic{\sigma}{j}) &= \overline{\$}\ . \qedhere
\end{align*}
\end{example}



\section{Discussion} \label{sec:discussion}

Although well-founded proof spaces are designed to prove termination of
parameterized concurrent programs, a natural question is how they relate to
existing methods for proving termination of sequential programs.  This section
investigates this question.  We will compare with the method of
\textbf{disjunctively well-founded transition invariants}, as exemplified by
Terminator \cite{Cook2005}, and the \textbf{language-theoretic} approach, as
used by Automizer \cite{conf/cav/HeizmannHP14}.


Terminator, Automizer, and our approach using well-founded proof spaces employ
the same high-level tactic for proving termination. The termination argument
is constructed incrementally in a sample-synthesize-check loop: first, sample
a lasso of the program, then synthesize a candidate termination argument
(using a ranking function for that lasso), then check if the candidate
argument applies to the whole program.  However, they are based on
fundamentally different proof principles.

\parpic[r]{
  \begin{minipage}{3cm}
    \small
    \texttt{i = pos()}\\
    \textbf{if}\texttt{(0 $\leq$ i $\leq$ 1):}\\
    \hspace*{0.25cm}\texttt{i = 2*i - 1}\\
    \hspace*{0.25cm}\textit{{\color{gray}// i is either 1 or -1}}\\
    \hspace*{0.25cm}\textbf{while}\texttt{(x$>$0 $\land$ z$>$0):}\\
    \hspace*{0.5cm}\texttt{x = x + i}\\
    \hspace*{0.5cm}\texttt{z = z - i}
\end{minipage}} Terminator is based on the principle of disjunctively well-founded
transition invariants.  Terminator proves termination by showing that the
\emph{transitive closure} of a program's transition relation is contained
inside the union of a finite number of well-founded relations.  As a concrete
example, consider the program to the right.  Assuming that we restrict
ourselves to linear ranking functions, well-founded proof spaces (and
Automizer) cannot prove that this program terminates, because there is no
linear term that decreases at every loop iteration.  Terminator can prove this
program terminates by showing that no matter how many iterations of the loop are
executed, \texttt{x} decreases \emph{or} \texttt{z} decreases.

\parpic[r]{
  \begin{minipage}{2.1cm}
    \small
\texttt{flag = true}\\
\textbf{while}\texttt{(z $>$ 0)}:\\
\hspace*{0.25cm}\textbf{if}\texttt{(flag):}\\
\hspace*{0.5cm}\texttt{x = z}\\
\hspace*{0.5cm}\texttt{z = pos()}\\
\hspace*{0.25cm}\textbf{else}\texttt{:}\\
\hspace*{0.5cm}\texttt{z = x - 1}\\
\hspace*{0.5cm}\texttt{x = pos()}\\
\hspace*{0.25cm}\texttt{flag = $\lnot$flag}
  \end{minipage}} Like well-founded proof spaces, Automizer is based on a language-theoretic
view of termination.  Automizer proves termination by exhibiting a family of
B\"{u}chi automata, each of which recognizes a language of traces that
terminate ``for the same reason'' (some given ranking function decreases
infinitely often), and such that every trace of the program is recognized by
one of the automata.  Assuming that we restrict ourselves to linear ranking
functions, Terminator cannot prove the program to the right terminates because
there is no linear disjunctively well-founded relation that includes the
\emph{odd} loop iterations.  Automizer (and well-founded proof spaces) can
prove the program terminates using the linear ranking function \texttt{z},
which decreases infinitely often along any infinite trace (at every even
loop iteration).

In the case of \emph{non}-parameterized concurrent programs, well-founded
proof spaces are equivalent in power to Automizer.  Suppose that $P$ is a
program that is intended to be executed by a fixed number of threads
$N$ (i.e., we are interested only proving that every trace in $P(N)$
terminates).  In this case, the premise of the proof rule ($\lang(P) \subseteq
\omega(\mathscr{H},\rankformulas)$) can be checked effectively using
algorithms for B\"{u}chi automata, due to the fact that both $\lang(P(N))$ and
$\omega(\mathscr{H},\rankformulas) \cap \iSigma{N}^\omega$ are
$\omega$-regular.

To cope with parameterized programs in which the number of threads is
arbitrary, Section~\ref{sec:check} describes a \emph{lasso} variation of the
proof rule (wherein we check $\$(\lang(P)) \subseteq
\$(\mathscr{H},\rankformulas)$ as a means to prove that $\lang(P) \subseteq
\omega(\mathscr{H},\rankformulas)$).  The lasso proof rule is strictly weaker
than Automizer's, and the above program cannot be verified for the same reason
that Terminator fails: there is no ranking function that decreases after odd
iterations of the loop.  That is, we cannot construct a well-founded proof
space such that $\$(\mathscr{H},\rankformulas)$ contains $\tau\$\rho^i$ for
odd $i$ (where $\tau$ represents the stem \texttt{flag = true} and $\rho$
represents one iteration of the \textbf{while} loop).  There is an interesting
connection between the lasso variant of well-founded proof spaces and
disjunctively well-founded transition invariants.  Terminator checks that the
transitive closure of the transition relation is contained inside a given
disjunctively well-founded relation by proving safety of a transformed
program.  The transformed program executes as the original, but (at some
point) non-deterministically saves the program state and jumps to another
(disconnected) copy of the program, in which at every loop iteration the
program asserts that the ``saved'' and ``current'' state are related by the
disjunctively well-founded relation.  Intuitively, this jump corresponds to
exactly the $\$$ marker in lasso languages: the traces that perform the jump
in the transformed program can be put in exact correspondence with the traces
of the lasso language $\$(\lang(P))$.


Thus, well-founded proof spaces relate to both the Terminator and Automizer
proof rules.  Section~\ref{sec:proofspace} is aligned with the
language-theoretic view of program termination used by Automizer.
Section~\ref{sec:check} mirrors the program transformation employed by
Terminator to cope with transitive closure.


\section{Conclusion} \label{sec:conclusion}

This paper introduces well-founded proof spaces, a formal foundation for
automated verification of liveness properties for parameterized programs.
Well-founded proof spaces extend the incremental termination proof strategy pioneered in
\cite{Cook2005,pldi/CookPR06} to the case of concurrent programs with unboundedly many
threads.  This paper investigates a
logical foundation of an automated proof strategy.  We leave for future work
the problem of engineering heuristic techniques to make the framework work in
practice.

\bibliographystyle{plain}
\bibliography{paper}

\appendix

\section{Proofs}

\thmsoundness*
\begin{proof}
  For a contradiction, suppose that there is an execution
  \[ s_0 \xrightarrow{\ic{\sigma_1}{\idx{i}_1}} s_1 \xrightarrow{\ic{\sigma_2}{\idx{i}_2}} \dotsi \]
  such that $\ic{\sigma_1}{\idx{i}_1}\ic{\sigma_2}{\idx{i}_2}\dotsi \in
  \omega(\mathscr{H},\rankformulas)$.  Then there is a sequence of naturals
  $\{\alpha_k\}_{k \in \mathbb{N}}$ and a ranking formula $w \in
  \rankformulas$ which satisfy the conditions of Definition~\ref{def:omega-h}.
  It is straightforward to show that $\{s_{\alpha_k}\}_{k \in \mathbb{N}}$ is an
  infinite descending sequence of program states in the order $w$,
  contradicting the fact that $w$ defines a well-founded order.
\end{proof}

\proppaprogram*
\begin{proof}
  Let $P = \tuple{\Loc,\Sigma,\ell_\init,\src,\tgt}$ be a parameterized
  program.

  Intuitively, we can recognize that $\tau\$\rho$ is a lasso of the program
  $P$ (reading $\tau\$\rho$ from right to left) as follows:
  \begin{itemize}
  \item While reading the loop $\rho$, we keep track of the control location
    of every thread, but also ``remember'' the control location in which each
    thread started.  This is accomplished with unary predicates of the form
    $\tuple{\ell_1,\ell_2}(i)$ (with $\ell_1,\ell_2 \in \Loc$), so that
    $\tuple{\ell_1,\ell_2}(i)$ holds when thread $i$ is at location $\ell_1$
    and started in location $\ell_2$
  \item When reading the separator symbol $\$$, we verify that each thread is
    in a loop by transitioning from $\tuple{\ell_1,\ell_2}(i)$ to $\ell_1(i)$
    if $\ell_1 = \ell_2$ (i.e., thread $i$ is currently at the same location
    it started in), and otherwise rejecting by transitioning to $\false$.
  \item When reading the stem $\tau$, we track the control location of each
    thread using unary predicates of the form $\ell(i)$.
  \item To verify that after reading $\tau\$\rho$ every thread is in the
    initial control location (equivalently, no thread is at a location other
    than the initial one) by making every predicate symbol except $\ell_\init$
    rejecting.
  \end{itemize}

  Formally, we construct the QPA $\mathcal{A}(P) =
  \tuple{Q,\ar,\Sigma,\delta,\phi_\start,F}$ as follows:
  \begin{itemize}
  \item $Q = \Loc \cup (\Loc \times \Loc) \cup \{ \Delta, \overline{\$}, \loc \}$.  Intuitively,
    \begin{itemize}
    \item $\Delta(i)$ stands for the disjunction $\bigvee_{\ell \in
      \Loc} \tuple{\ell,\ell}(i)$
    \item $\loc(i)$ stands for the disjunction $\bigvee_{\ell\in \Loc} \ell(i)$
    \item $\overline{\$}$ is used to enforce the condition that the loop of a
      lasso may not be empty (i.e., $\$$ does not appear at the end of an
      accepted word)
    \end{itemize}
      
  \item The arity of each predicate symbol is 1, except $\overline{\$}$ which has
    arity 0.
  \item $\phi_\start = \overline{\$} \land \forall i. \Delta(i)$ (i.e.,
    initially every thread is currently where it started)
  \item $F = \{ \ell_\init, \loc \}$
  \end{itemize}
  The transition function $\delta$ is defined as follows: {\small

    \vspace*{1pt}
    \noindent$\delta(\tuple{\ell_1,\ell_2}(i), \ic{\sigma}{j}) =$\\
    \null\hfill$\begin{cases}
      \ite{i=j}{\tuple{\src(\sigma),\ell_2}(i)}{\tuple{\ell_1,\ell_2}(i)} & \text{if } \tgt(\sigma) = \ell_1\\
      i \neq j \land \tuple{\ell_1,\ell_2}(i) &\text{otherwise}
    \end{cases}$\\

    \noindent$\delta(\ell(i), \ic{\sigma}{j}) =$\\
    \null\hfill$\begin{cases}
      \ite{i=j}{\src(\sigma)(i)}{\ell(i)} & \text{if } \tgt(\sigma) = \ell\\
      i \neq j \land \ell(i) &\text{otherwise}
    \end{cases}$
  \begin{align*}
    \delta(\Delta(i), \ic{\sigma}{j}) &= \ite{i=j}{\tuple{\src(\sigma),\tgt(\sigma)}(i)}{\Delta(i)}\\
    \delta(\loc(i), \ic{\sigma}{j}) &= \ite{i=j}{\src(\sigma)(i)}{\loc(i)}\\
    \delta(\overline{\$}, \ic{\sigma}{j}) &= \true\\
    \delta(\tuple{\ell_1,\ell_2}(i), \$) &= \begin{cases}
      \ell_1(i) & \text{if } \ell_1 = \ell_2\\
      \false &\text{otherwise}
    \end{cases}\\
    \delta(\ell(i), \$) &= \false\\
    \delta(\Delta(i), \$) &= \loc(i)\\
    \delta(\loc(i), \$) &= \false\\
    \delta(\overline{\$}, \$) &= \false
  \end{align*}
  }
\end{proof}

\proppaproof*
\begin{proof}
  Our construction follows a similar structure to the construction of a
  (quantifier-free) predicate automaton from a proof space \cite{Farzan2015}.
  Let $\tuple{\mathscr{H},\rankformulas}$ be a regular well-founded proof space, and let $\tuple{H,W}$ be a basis for $\tuple{\mathscr{H},\rankformulas}$.

  The intuition behind the construction of a QPA $\mathcal{A}(H,W)$ which
  recognizes $\$(\mathscr{H},\rankformulas)$ is that the predicate symbols in $A$
  correspond to program assertions in $H$, and the transition function
  corresponds to the Hoare triples in $H$.  More explicitly, we define a QPA
  $\mathcal{A}(H,\rankformulas) = \tuple{Q,\ar,\Sigma,\delta,\phi_\start,F}$ as follows.

  The set of predicate symbols $Q$ is the set of \emph{canonical names} for
  the assertions which appear in $H$.  A canonical name is a representation of
  an equivalence class of program assertions, where two assertions $\phi$ and
  $\psi$ are equivalent if there is a permutation of thread identifiers $\pi :
  \mathbb{N} \rightarrow \mathbb{N}$ so that $\phi[\pi] = \psi$.  For example,
  the assertions $\texttt{m}(4)<\texttt{m}(2)$ and
  $\texttt{m}(2)<\texttt{m}(9)$ are both represented by the same canonical
  assertion, which we write as $[\texttt{m}(2)<\texttt{m}(9)]$.  The arity of
  a predicate symbol is the number of distinct thread indices which appear in it
  (e.g., $\ar([\texttt{m}(2)<\texttt{m}(9)]) = 2$).

  Each Hoare triple in $H$ corresponds to a transition rule of
  $\mathcal{A}(H,W)$.  For example, the Hoare triple
  \[\hoare{\texttt{m}(1)<\texttt{t}}{\ic{\texttt{m=t++}}{2}}{\texttt{m}(1)<\texttt{m}(2)}\]
  corresponds to the transition
  \[\delta([\texttt{m}(1)<\texttt{m}(2)](i,j), \ic{\texttt{m=t++}}{k}) = k = j \land [\texttt{m}(1)<\texttt{t}](i) \]
  If there are multiple Hoare triples with the same command and canonical
  post-condition, then the transition rules are combined via disjunction.  For
  example, if the following Hoare triple also belongs to the basis:
  \[\hoare{\texttt{m}(2)<\texttt{m}(1)}{\ic{\texttt{m=t++}}{3}}{\texttt{m}(2)<\texttt{m}(1)}\]
  then the transition rule is:\\
  $\delta([\texttt{m}(1)<\texttt{m}(2)](i,j), \ic{\texttt{m=t++}}{k}) = k = j \land [\texttt{m}(1)<\texttt{t}](i)$\\
  \null\hfill$\lor k \neq i \land k \neq j \land [\texttt{m}(1)<\texttt{m}(2)](i,j)$

  For any global variable $g$, by reading $\$$ we may transition from
  $[\old{g} = g]$ to $\true$ (and similarly for local variables $l$):
  \begin{align*}
  \delta([\old{g} = g],\ic{\$}{i_0}) &= \true\\
  \delta([\old{\iv{l}{1}} = \iv{l}{1}](i_1),\ic{\$}{i_0}) &= \true\\
  \end{align*}
  For all other predicate symbols $q$, reading $\$$ has no effect:
  \[ \delta(q(i_1,...,i_{\ar(q)}),\ic{\$}{i_0}) = q(i_1,...,i_{\ar(q)})\]

  The initial formula of $\mathcal{A}(H,W)$ expresses the desired
  post-condition that the \emph{some} ranking formula decreases.  Formally, 
  $\phi_\init$ is defined as
  follows:
  \[ \phi_\init = \bigvee_{w \in W} \exists \vec{i}. w(\vec{i})\]
  Lastly, there are no accepting predicate symbols ($F=\emptyset$), which expresses
  the desired pre-condition $\true$.
\end{proof}

\propqpaclosure*
\begin{proof}
  Let $A$ and $A'$ be PAs.  We form their intersection $A \cap A'$ by taking
  the vocabulary to be the disjoint union of the vocabularies of $A$ and $A'$,
  and define the transition relation and accepting predicates accordingly.
  The initial formula is obtained by conjoining the initial formulas of $A$
  and $A'$.  The union $A \cup A'$ is formed similarly, except the initial
  formula is the disjunction of the initial formulas of $A$ and $A'$.

  Given a PA $A = \tuple{Q,\ar,\Sigma,\delta,\phi_\init,F}$, we form its
  complement $\overline{A} =
  \tuple{\overline{Q},\overline{\ar},\Sigma,N,\overline{\delta},\overline{\phi_\init},\overline{F}}$
  as follows.  We define the vocabulary $(\overline{Q},\overline{\ar})$ to be a
  ``negated copy'' of $(Q,\ar)$: $\overline{Q} = \{ \overline{q} : q \in
  Q \}$ and $\overline{\ar}(\overline{q}) = \ar(q)$.  The set of accepting
  predicate symbols is the (negated) set of rejecting predicate symbols
  from $A$: $\overline{F} = \{ \overline{q} \in \overline{Q} : q \in Q
  \setminus F \}$.  For any formula $\phi$ in $\formulae(Q,\ar)$ in the vocabulary of $A$, we use
  $\overline{\phi}$ to denote the ``De Morganization'' of $\phi$ in the vocabulary of $\overline{A}$, defined
  recursively by:
  \begin{center}
  $\overline{q(i_{j_1},...,i_{j_{\ar(q)}})} = \overline{q}(i_{j_1},...,i_{j_{\ar(q)}})$\\
    \begin{minipage}[t]{3cm}
      \vspace*{-5pt}
    \begin{align*}
      \overline{i_j = i_k} &= (i \neq j)\\
      \overline{\phi \land \psi} &= \overline{\phi} \lor \overline{\psi}\\
      \overline{\exists i. \phi} &= \forall i. \overline{\phi}
    \end{align*}
  \end{minipage}
  \begin{minipage}[t]{3cm}
      \vspace*{-5pt}
    \begin{align*}
      \overline{i_j \neq i_k} &= (i = j)\\
      \overline{\phi \lor \psi} &= \overline{\phi} \land \overline{\psi}\\
      \overline{\forall i. \phi} &= \exists i. \overline{\phi}
    \end{align*}
  \end{minipage}
  \end{center}

  We define the transition function and initial formula of $\overline{A}$ by
  De Morganization: $\overline{\delta}(\overline{q},\sigma)$ is defined to be
  $\overline{\delta(q,\sigma)}$ and the initial formula is defined to be
  $\overline{\phi_\init}$.
\end{proof}

\thmcert*
\begin{proof}
  Let $A = \tuple{Q,\ar,\Sigma,\delta,\phi_\start,F}$ be a QPA and let $\phi$
  be an emptiness certificate for $A$.  Intuitively, the \emph{Initialization}
  and \emph{Consecution} conditions for $\phi$ express that $\phi$ is an
  inductive invariant for the transition system on $A$-configurations, while
  the \emph{Rejection} condition expresses that no model of $\phi$ is
  accepting.
  
  The non-standard part of argument is proving the soundness of the
  consecution condition, since the consecution condition expresses that $\phi$
  is inductive under $\hat{\delta}$ rather than the transition relation on
  $A$-configurations.  In other words, we must establish that if
  $\config,\config'$ are $A$-configurations, $\sigma \in \Sigma$ and $\idx{k}
  \in \mathbb{N}$ so that that $\config \models \phi$ and
  $\ctrans{\config}{\sigma}{\idx{k}}{\config'}$, then $\config' \models
  \hat{\delta}(\phi, \sigma)$.  This property can be proved by induction on
  $\phi$.
\end{proof}

\thmqltl*
\begin{proof}
  Any $\QLTL(\Sigma)$ formula is equivalent to one written as a disjunction of
  $\QLTL(\Sigma)$ formulas of the form \[\mathcal{Q}_1 i_1. \mathcal{Q}_2
  i_2,... \mathcal{Q}_k i_{k}. \big(\bigwedge_{j\neq j'} i_j \neq i_{j'}\big)
  \land \phi_m(i_1,...,i_k)\ ,\] where each $\mathcal{Q}_j$ is either
  $\exists$ or $\forall$, and $\phi_m(i_1,...,i_k)$ is quantifier-free.  It is
  sufficient to construct the QPA for a single disjunct of this form, since
  QPA languages are closed under union (Proposition~\ref{prop:qpa-closure}).

  Since the formula $\phi_m[i_1 \mapsto
    1, \ldots, i_k \mapsto k]$ is equivalent to an LTL formula, the
  the set of all \itrace{}s $\tau \in \Sigma(k+1)^\omega$ such that
  \[\tau \models \phi_m[i_1 \mapsto 1, \ldots, i_k \mapsto k]\] can be
  recognized by a B\"{u}chi automaton \cite{Vardi1994}.  Let
  $A^\omega(\phi_m)$ denote this B\"{u}chi automaton.  From
  $A^\omega(\phi_m)$, we may derive a deterministic finite automaton
  $A^\$(\phi_m) = \tuple{Q_m,\iSigma{k+1},\delta_m,s_m,F_m}$ which recognizes
  $\$(\lang(A^\omega(\phi)))$, following the construction from
  \cite{Calbrix1993}.  From $A^\$(\phi_m)$, we may construct a
  QPA \[\mathcal{A}(\phi) = \tuple{Q,\ar,\Sigma,\delta,\phi_\start,F}\] which
  recognizes the language
  \[\$(\lang(\mathcal{Q}_1 i_1. \mathcal{Q}_2 i_2,... \mathcal{Q}_k
  i_k. \big(\bigwedge_{j\neq j'} i_j \neq i_{j'}\big) \land
  \phi_m(i_1,...,i_k)))\] as follows:
  \begin{itemize}
  \item $Q = Q_m$
  \item For each predicate symbol $q \in Q$, define $\ar(q) = k$
  \item $\phi_\init = \mathcal{Q}_1 i_1. \dotsi \mathcal{Q}_k i_k. \big(\bigwedge_{j\neq j'} i_j \neq i_{j'}\big) \land \big(\bigvee_{q \in F_m} q(i_1,...,i_k)\big)$
  \item For every $q \in Q_m$ and $\sigma \in \Sigma \cup \{\$\}$, $\delta$ is defined by:
    {\small
      $\delta(q(\vec{i}), \ic{\sigma}{i_0}) =$\hfill$i_0 = i_1 \land \big(\bigvee \{ q'(\vec{i}) : \delta_m(q',\ic{\sigma}{1}) = q \}\big)$\\
      \null\hfill$\vdots$\hfill\null\\
      \null\hfill$\lor i_0 = i_k \land \big(\bigvee \{q'(\vec{i}) : \delta_m(q',\ic{\sigma}{\idx{k}}) = q \}\big)$\\
      \null\hfill$\lor \bigwedge_{j = 1}^k i_0 \neq i_j \land \big(\bigvee \{q'(\vec{i}) : \delta_m(q',\ic{\sigma}{\idx{k}+1}) = q \}\big)$
    }
  \item $F = \{q_0\}$ \qedhere
  \end{itemize}
\end{proof}

\end{document}